\documentclass[10pt,twocolumn,twoside]{IEEEtran}

\usepackage{graphicx} 
\usepackage{epsfig} 
\usepackage{times} 
\usepackage{amsmath} 
\usepackage{amssymb}  
\usepackage{mathdots}
\usepackage{mathtools}
\usepackage{wrapfig}
\usepackage{amsthm}
\usepackage{caption}
\usepackage{cite}
\usepackage{comment}
\usepackage{subfigure}
\usepackage{float}

\newtheorem{thm}{Theorem}
\newtheorem{lem}{Lemma}
\newtheorem{assumption}{Assumption}
\newtheorem{remark}{Remark}

\title{Auxiliary Network-Enabled Attack Detection and Resilient Control of Islanded AC Microgrid}
\author{Vaibhav~Vaishnav,~\IEEEmembership{Student Member,~IEEE},~Anoop~Jain,~\IEEEmembership{Member,~IEEE}, and Dushyant~Sharma,~\IEEEmembership{Member,~IEEE}
\thanks{V. Vaishnav and A. Jain are with the Department of Electrical Engineering, Indian Institute of Technology Jodhpur 342030, India (e-mail:  {vaishnav.2@iitj.ac.in};  {anoopj@iitj.ac.in}). D. Sharma is with Department of Electrical Engineering, Indian Institute of Technology (Indian School of Mines) Dhanbad 826004, India (e-mail:  {dushyant@iitism.ac.in}).}}

\begin{document}

\maketitle

\begin{abstract}
This paper proposes a cyber-resilient distributed control strategy equipped with attack detection capabilities for islanded AC microgrids in the presence of bounded stealthy cyber attacks affecting both frequency and power information exchanged among neighboring distributed generators (DGs). The proposed control methodology relies on the construction of an auxiliary layer and the establishment of effective inter-layer cooperation between the actual DGs in the control layer and the virtual DGs in the auxiliary layer. This cooperation aims to achieve robust frequency restoration and proportional active power-sharing. It is shown that the in situ presence of a concealed auxiliary layer not only guarantees resilience against stealthy bounded attacks on both frequency and power-sharing but also facilitates a network-enabled attack identification mechanism. The paper provides rigorous proof of the stability of the closed-loop system and derives bounds for frequency and power deviations under attack conditions, offering insights into the impact of the attack signal, control and pinning gains, and network connectivity on the system's convergence properties. The performance of the proposed controllers is illustrated by simulating a networked islanded AC microgrid in a Simulink environment showcasing both attributes of attack resilience and attack detection.
\end{abstract}

\begin{IEEEkeywords}
AC microgrids, cyber-security, network-enabled attack detection, resilient control, stealthy attacks.
\end{IEEEkeywords}


\section{Introduction}\label{Introduction}
Distributed control has been the most widely adopted strategy over the past decade to exercise frequency and voltage control of islanded AC microgrids \cite{zhou2020distributed}. The cooperation among the DGs in a microgrid relies on local information exchange over a sparse communication network, which is vulnerable to cyber attacks \cite{peng2019}. Such attacks, commonly known as false data injection (FDI) attacks, jeopardize the integrity of information transmitted across communication channels and have garnered significant attention in resilient microgrid control \cite{reda2022}. A trust factor-based control regime was proposed in \cite{lu2020distributed} to synchronize the frequencies of DGs under attack, by real-time monitoring of a confidence factor associated with all the neighboring DGs. Wang et al.\cite{wang2020cyber} proposed cyber-resilient adaptive controllers for frequency and voltage restoration with active/reactive power-sharing in networked AC/DC microgrids. These controllers utilize time-varying adaptive gains based on frequency, voltage, and power errors. An auxiliary state was introduced in \cite{sadabadi2021} to make conventional distributed frequency consensus protocol cyber resilient and maintain frequency synchronization and proportional active power-sharing. Recently, \cite{xiao2023} proposed attack magnitude-dependent time-variant communication weights, to implement an adaptive frequency regulation strategy using the most optimal communication channels in the presence of multiple link attacks.

Apart from incorporating cyber resilience in the control framework, topology switching induced by different types of attack detection techniques has been the other counterpart of attack mitigation strategies in microgrid control. These techniques can be further classified into (a) system-based \cite{intriago2023,madichetty2021,liu2022} and (b) data-based approaches \cite{wan2022,takiddin2022}, where the states of system under consideration are continuously estimated using an observer in (a), while historical training data is used for attack detection for redundancy in (b). To the best of the authors' knowledge, the concepts of \emph{attack resilience} and \emph{attack detection} have often been addressed separately in the microgrid control literature. Therefore, designing a suitable distributed control framework together with resilience and detection capabilities is a crucial research direction for ensuring the safety and security of microgrid control.

In recent times, cyber adversaries have evolved to employ increasingly sophisticated and deceptive stealthy attacks \cite{zhang2021stealthy} that might be challenging to ascertain using novel detection strategies like \cite{intriago2023,madichetty2021,liu2022,wan2022,takiddin2022}. Equipped with the complete knowledge of dynamics and control architecture of the system, the attacks are now specifically designed as bounded intermittent disturbances with state-dependent linear/nonlinear dynamics \cite{zhang2021stealthy,lu2023concurrent} which can easily deceive state observers and showcase zero neighborhood error, even when the neighboring DGs are out of synchronization \cite{mustafa2019}. In view of such bounded stealthy attacks, a virtual hidden layer-based cyber resilient framework was proposed in \cite{gusrialdi2018} for distributed control of networked systems. Motivated by \cite{gusrialdi2018}, in this paper, we propose auxiliary virtual network-enabled distributed resilient control for frequency synchronization and proportional active power-sharing among a group of inverter-based DGs. 

Distributed controllers adopting the aforementioned idea of the virtual layer have been previously developed in \cite{chen2020,zhou2021,liu2021robust,jamali2023}, addressing the problem of frequency synchronization in microgrids. However, these works do not emphasize the need for a dedicated power controller that can provide resilience against any possible attack on power-sharing between neighboring DGs. The distributed secondary controllers in \cite{chen2020,zhou2021,liu2021robust,jamali2023} are devised to provide a frequency set point to the droop-based primary control which depends on the inherent linear relationship between frequency and active power \cite{guerrero2010}. However, any FDI attack affecting the sharing of active power can significantly hamper system frequency and lead to deviations in power demand beyond the rated values of the respective DGs. Therefore, a resilient power controller is obligatory for preserving the resilience of the frequency controller against cyber attacks. Motivated by this fact, in this paper, we design resilient frequency and power controllers capable of functioning in the presence of simultaneous attacks on both frequency and active power within a microgrid.

It is also evident to mention that the stability under proposed controllers in \cite{gusrialdi2018,chen2020,liu2021robust,jamali2023} is leveraged by assumptions on the magnitude of resilient control gain and inclusion of attack vector in the Lyapunov candidate function. However, considering that a cyber attack is an external breach that is completely unknown to both the system and the control engineer, establishing the closed-loop stability of a cyber-resilient control system relying solely on the state dynamics of the system appears to be a more practical and mathematically viable approach. Further, the control architecture inclusive of an auxiliary layer hidden from the attacker provides an alternate safeguarded path for communication among neighboring DGs, thereby giving a provision to compare and detect any discrepancy in information being shared through the actual cyber layer. Leveraging this, we propose a straightforward test for detecting an attack, as compared to the works \cite{chen2020,zhou2021,liu2021robust,jamali2023}. In the light of above discussion, the major contributions of this paper can be summarized as follows:

\begin{enumerate}
	\item[i)] Relying on an auxiliary layer, we propose frequency and power-distributed controllers as leader-follower and leaderless resilient cooperative systems, respectively, for frequency regulation and proportional active power-sharing in an islanded AC microgrid, under the presence of FDI attacks affecting both frequency and power-sharing between the neighboring DGs. The FDI attacks are uniformly bounded and governed by the dynamics having finite $\mathcal{L}_2$ gain.  
\item[ii)] Due to the difference in the nature of cooperation and distributed consensus protocols, we provide separate stability analyses under frequency and power attacks and obtain analytical bounds characterizing the deviations of frequency and power from their steady-state value under no-attack conditions. It is shown that these bounds depend on the underlying network topology, the magnitude of the attack signals, pinning gain and a gain term decided by the control designer.   
\item[iii)] As a byproduct of auxiliary layer-based control design, we propose a network-enabled attack detection mechanism showcasing secured propagation of control variables between the cyber and auxiliary layer, thereby devising a detection test for surveillance of any possible attack in the cyber layer.    
\end{enumerate}

Additionally, an extensive simulation study is presented under the simultaneous presence of both cyber attacks and load perturbations. The remnant of this paper is arranged as follows: Section \ref{Introduction} commences with an introduction, summarizing mathematical preliminaries and a brief review of secondary droop control. Section \ref{system_description} describes system architecture and attack model and formulates the problem. Section \ref{Frequency_control} proposes a resilient frequency controller and obtains various bounds characterizing the effect of attack and network connectivity on the steady-sate frequency error. The philosophy behind resilient power controllers is discussed under Section \ref{Power_control}. Section \ref{Attack_detection} explains the auxiliary network-enabled attack identification scheme and proposes a generalized test for detecting the attacks.  Lastly, Section \ref{Simulation_case_studies} presents simulation case studies, before concluding the paper in Section \ref{Conclusion}.

\subsection{Notations, Graph Theory, and Mathematical Preliminaries}\label{Preliminaries}
Throughout the paper, $\mathbb{R}$ denotes the set of real numbers, and $\textbf{{0}}_{n} = [0,\ldots,0]^T \in \mathbb{R}^n, \textbf{{1}}_{n} = [1,\ldots,1]^T \in \mathbb{R}^n$. For $x = [x_1, \ldots, x_n]^T \in \mathbb{R}^n$, $\|x\|$ represents its Euclidean norm. The complex number is represented by $j_c = \sqrt{-1}$. The symbol $\otimes$ denotes the Kronecker product of two matrices. $\lambda(M)$ represents an eigenvalue of any square matrix $M \in \mathbb{R}^{n \times n}$, $M^T$ is its transpose and $\|M\|$ denotes its operator (or induced) 2-norm such that $\|M\| = \sqrt{\lambda_{\max}(M^TM)}$, where $\lambda_{\max}$ is the maximum eigenvalue of the symmetric matrix $M^TM$. Consider a network of $n$ DGs, described as an undirected graph $\mathcal{G}=(\mathcal{V}, \mathcal{E})$, with node set $\mathcal{V} = \{1, \ldots, n\}$ and edge set $\mathcal{E} \subset \mathcal{V} \times \mathcal{V}$. The information flow between $i^\text{th}$ and $j^\text{th}$ nodes can be represented by an undirected edge $(i,j) \in \mathcal{E} \Leftrightarrow (j,i) \in \mathcal{E}$. The set of neighbors of the $i^\text{th}$ node is denoted by $\mathcal{N}_i$ = $\{j:(i,j) \in \mathcal{E}, j \neq i\}$. The adjacency matrix of $\mathcal{G}$ is denoted by $\mathcal{A}$ = $[a_{ij}] \in \mathbb{R}^{n\times n}$ with $a_{ij} = 1 \Leftrightarrow (i,j) \in \mathcal{E}$, else $a_{ij} = 0$. The Laplacian matrix of $\mathcal{G}$ is denoted by  $L = [\ell_{ij}] \in \mathbb{R}^{n\times n}$, where $\ell_{ii} = \mathcal{\sum}_{j \in \mathcal{N}_i} a_{ij}$ and $\ell_{ij} = - a_{ij}$ for $i \neq j$. For an undirected and connected graph $\mathcal{G}$, the eigenvalues of the Laplacian $L \in \mathbb{R}^{n \times n}$ can be arranged in the ascending order as $0 = \lambda_1(L) \leq \lambda_2(L) \leq \cdots \leq \lambda_n(L)$. The following lemmas will be useful in the subsequent analysis of this paper: 

\begin{lem}[\hspace{-1pt}\cite{wang2010}, \cite{vaishnav2023}]\label{lem_1}
Let $\mathcal{G}$ be an undirected and connected graph with Laplacian $L \in \mathbb{R}^{n \times n}$, and $B = {\rm diag}\{b_1, \ldots, b_n\} \in \mathbb{R}^{n \times n}$ be a diagonal matrix with diagonal entries $b_i > 0, \forall i$. Then, the matrix $L + B$ is positive-definite.
\end{lem}

\begin{lem}[\hspace{-1pt}\cite{dong2020resilient}]\label{lem_2}
Consider a uniformly bounded signal $\|p(t)\| \leq \Omega, \ \forall t \geq 0$, where $\Omega$ is a positive constant. Then, for a Hurwitz matrix $Q$, there exists a constant vector $\Psi$ and some finite-time $T$ such that following holds true for all $t \geq T$:
\begin{equation*}
\left\|\int_{0}^{t} {\rm e}^{Q(t- \tau)}p(\tau) d\tau \right\| \leq \left\| \int_{0}^{t} {\rm e}^{Q(t- \tau)} \Psi d\tau \right\|.
\end{equation*}
\end{lem}

\subsection{Distributed Secondary Droop Control}\label{droopcontrol}
Droop control exploits the linear dependence of frequency on active power, causing the frequency of an inverter-based DG to droop with its output active power \cite{guerrero2005}, according to the relation
\begin{equation}\label{droop}
	\omega_i = \omega_{o_i} - m_{P_i}P_i,
\end{equation}
where $\omega_i$ is angular frequency, $\omega_{o_i}$ is the nominal angular frequency, $P_i$ is measured output active power and $m_{P_i}$ is the droop coefficient associated with the $i^\text{th}$ DG. Due to primary droop control, since the inverter's frequency undergoes a deviation from its nominal value, a suitable secondary control needs to be designed to restore the frequency of all the inverters to nominal microgrid frequency \cite{guerrero2010}. This is realized by cooperation among all the DGs at the secondary control level and designing distributed controllers $u_{\omega_i} = \dot{\omega}_i$ and $u_{P_i} = m_{P_i}\dot{P}_i$ for frequency and active power for each $i$. Using these, the nominal frequency is obtained from \eqref{droop} as:   
\begin{equation}\label{nominal_frequency}
	\dot{\omega}_{o_i} = \dot{\omega}_i + m_{P_i}\dot{P}_i \implies \omega_{o_i} = \int u_{\omega_i}~dt+ \int u_{P_i}~dt, \ \forall i.
\end{equation}

In the subsequent analysis, we use the vectors $\omega = [\omega_1, \ldots, \omega_n]^T \in \mathbb{R}^n$, $P = [P_1, \ldots, P_n]^T \in \mathbb{R}^n$ and $m_P P = [m_{P_1} P_1, \ldots, m_{P_n}P_n]^T \in \mathbb{R}^n$ to represent the frequency, power, and droop-coefficient associated active power for the network of $n$ DGs. 

\begin{figure}[t]
	\centering
	\includegraphics[scale=0.13]{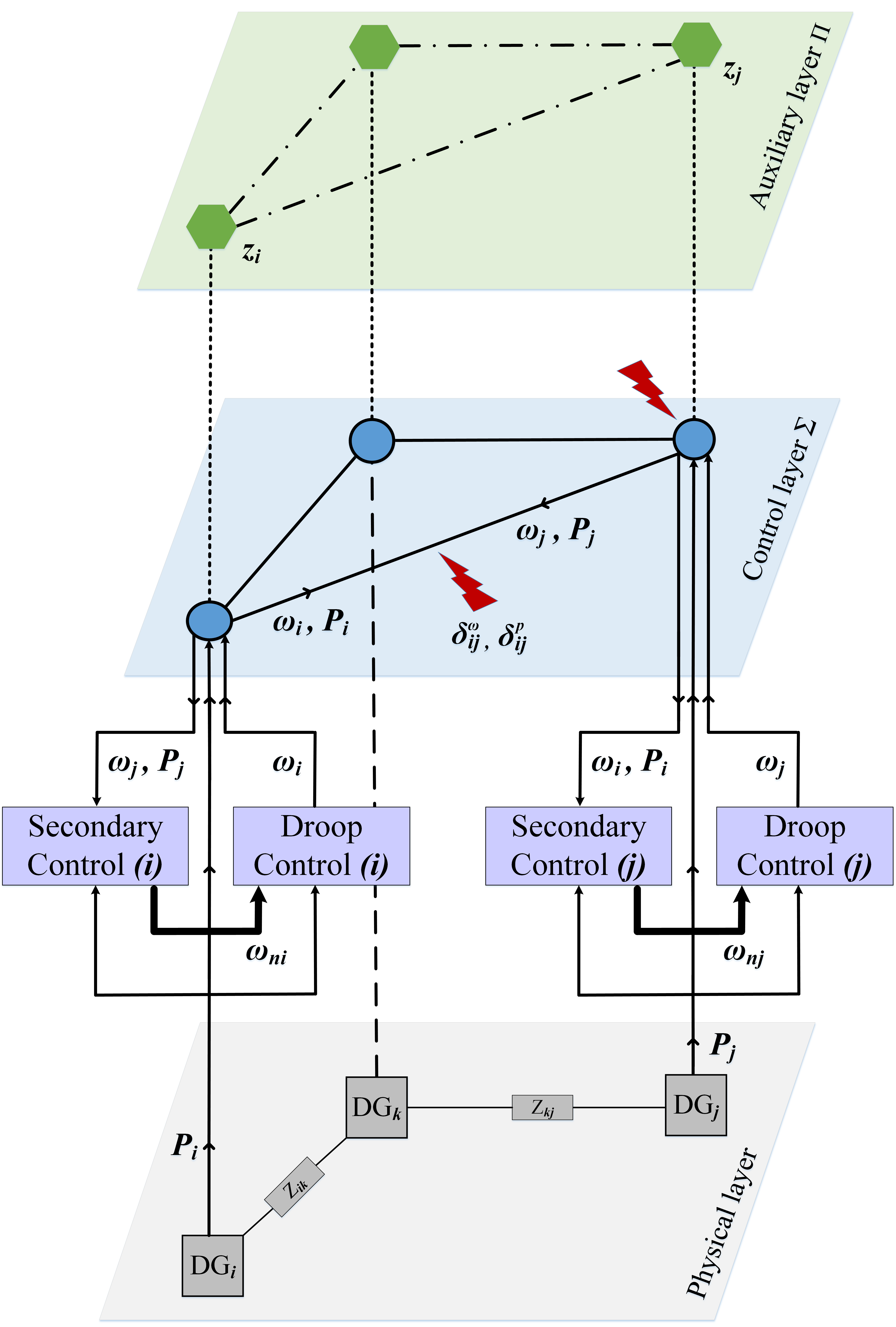}
	\caption{Microgrid network and control architecture.}
	\label{control_digram}
	\vspace*{-13pt}
\end{figure}

\section{System Description and Problem Statement}\label{system_description}

\subsection{System Description}
Consider a microgrid network comprising $n$ DGs in the physical layer, along with the associated control layer $\Sigma$ and the auxiliary layer $\Pi$, as shown in Fig.~\ref{control_digram}. Each DG is equipped with individual primary and secondary controllers where frequency $\omega_i$ and active power $P_i$ are the physical state variables to be controlled for each $i$. According to \eqref{nominal_frequency}, the secondary control generates nominal frequency to the primary droop controller for computing the actual frequency $\omega_i$ of the $i^\text{th}$ DG. Consequently, both frequency and the average active power signals of the DGs are shared among them according to the communication network at the control layer $\Sigma$ for designing frequency and power controllers. 

However, owing to the attacker's knowledge about the control layer, the communication network over $\Sigma$ is vulnerable to FDI attacks on both frequency and power signals, deviating the system from the desired cooperative control objectives. To account for these attacks, our approach relies on the construction of an auxiliary layer $\Pi$ and its integration with $\Sigma$ to assure resiliency toward frequency regulation and proportional active power-sharing. The auxiliary nodes in $\Pi$ are considered to have their own state dynamics but the same inter-agent interaction network as in $\Sigma$. Further, there exists a hidden network connecting the corresponding nodes in $\Sigma$ and $\Pi$ (see Fig.~\ref{control_digram}). 

Under an FDI attack, the frequency and power signals received at a particular node in $\Sigma$ may be different from the one transmitted by the neighboring nodes and can be expressed as:  
\begin{equation}\label{after_attack}
	\check{\omega}_{ij} = {\omega}_j + \delta_{ij}^{\omega}, \quad \check{P}_{ij} = {P}_j + \delta_{ij}^{P}, \quad j \in \mathcal{N}_i \cup \{i\}, 
\end{equation}
where $\check{\omega}_{ij}$ is the frequency signal received at node $i$ from node $j$, ${\omega}_j$ is the actual frequency of the $j^{\text{th}}$ DG, and $\delta_{ij}^{\omega}$ is the FDI attack affecting the link connecting nodes $i$ and $j$. Similar notations can be defined for $\check{P}_{ij}$. The total external injection at the $i^{\text{th}}$ node can be written as $\sum_{j \in \mathcal{N}_i \cup \{i\}} \delta_{ij}^{\omega} = d_{\omega_i}$ (resp., for power as $\sum_{j \in \mathcal{N}_i \cup \{i\}} \delta_{ij}^{P} = d_{P_i}$), leading to the attack vector $d_{\omega} = [d_{\omega_1},\ldots,d_{\omega_n}]^T$ (resp., $d_{P} = [d_{P_1}, \ldots, d_{P_n}]^T$) for the $n$ nodes in $\Sigma$. In general, an attacker might employ an FDI attack varying as per the real-time states of the system, which dies out as the system states become stable. Such kinds of attacks deviate the system from achieving the desired convergence properties and even make it unstable while being stealthy \cite{gusrialdi2018}. In this paper, we model the frequency attack (resp., power attack) as a function of states $\omega_i$ (resp., $P_i$) and time $t$ such that $d_{\omega}(\omega,t)$ (resp., $d_{P}(P,t)$) satisfies the following assumption:
\begin{assumption}\label{assumption}
The injections $d_{\omega}(\omega, t)$ and $d_{P}(P, t)$ are uniformly bounded for bounded $\omega$ and $P$, that is, there exist positive constants $\bar{D}_{\omega}, \bar{D}_{P}$ such that $\|d_{\omega}(\omega, t)\| \leq \bar{D}_{\omega}$ and $\|d_{P}(P, t)\| \leq \bar{D}_{P}$ for all $\omega \in \mathbb{R}^n, P \in \mathbb{R}^n$ and $t \geq 0$. In particular, if these attacks are generated by the dynamics,
\begin{equation}\label{attack_dynamics} 
\dot{d}_{\omega} = f_{\omega}(d_{\omega},\omega); \quad \dot{d}_P = f_P(d_P, P), 	
\end{equation}
these must have a finite $\mathcal{L}_2$-gain. 
\end{assumption}	

The relevance of this assumption lies in the fact that an intelligent attacker always aims at inserting a bounded attack signal, as the injections with large magnitude are easy to detect and will be rejected before they spread across the network. Further, dynamics \eqref{attack_dynamics} having finite $\mathcal{L}_2$-gain implies that $d_{\omega}$ and $d_P$ settle down to some value, as $\omega$ and $P$ reach the steady-state, respectively.

\subsection{Problem Statement}\label{Problem_formulation}
Consider $n$ DGs, connected via an undirected and connected graph $\mathcal{G}$, as shown in Fig.~\ref{control_digram}. Assume that the DGs be governed by the droop-based frequency dynamics \eqref{nominal_frequency} and subject to frequency and power attacks $d_{\omega}$ and $d_P$ in control layer $\Sigma$, according to Assumption~\ref{assumption}. Let $z = [z_1, \ldots, z_n]^T \in \mathbb{R}^n$ be the state vector associated with the auxiliary nodes and $\omega ^{*} > 0$ be the desired value of the microgrid frequency. Based on the integration between $\Sigma$ and $\Pi$ layers, design the vector functions $\mathcal{F}^{c}_{\omega}(\omega, z), \mathcal{F}^{a}_{\omega}(\omega, z), \mathcal{F}^{c}_{P}(P, z), \mathcal{F}^{a}_{P}(P, z)$ such that the following auxiliary-state coupled dynamics, under the frequency and power attacks $d_{\omega}$ and $d_P$, disseminated over $\Sigma$ using Laplacian $L$, respectively,
\begin{equation}\label{coupled_model}
\hspace*{-0.3cm} S_{\omega}:
	\begin{dcases*}
		\dot{\omega} & = $\mathcal{F}^{c}_{\omega}(\omega, z) + L d_{\omega}$\\
		\dot{z} &= $\mathcal{F}^{a}_{\omega}(\omega, z)$
	\end{dcases*},
	\
	S_P:
	\begin{dcases*}
	m_P\dot{P} &= $\mathcal{F}^{c}_{P}(P, z) + L d_{P}$\\
	\dot{z} &= $\mathcal{F}^{a}_{P}(P, z)$
	\end{dcases*},
\end{equation}
assures that the following hold true for some small constants $\epsilon_{\omega} > 0$ and $\epsilon_{P} > 0$:
\begin{itemize}
\item[{\bf (P1)}] Frequency is maintained at the desired nominal value for all the DGs, that is,  $\lim_{t \to \infty} \|\omega(t) - \omega^{*}\pmb{1}_n\| \leq \epsilon_{\omega}$.
\item[{\bf (P2)}] Power is shared among the DGs in proportion to their droop coefficients, that is, $\lim_{t \to \infty} \|m_{P} P - \Delta_P \pmb{1}_n\| \leq \epsilon_{P}$, where $\Delta_P = (1/n)\sum_{i=1}^{n} m_{P_i}P_i(0)$ is a constant. 
\end{itemize}
Additionally,
\begin{itemize}
\item[{\bf (P3)}] Based on the integration of $\Sigma$ and $\Pi$, propose an attack detection mechanism such that non-zero FDI attacks $\delta_{ij}^{\omega}$ and $\delta_{ij}^{P}$, as defined in \eqref{after_attack}, are detected for all $(i, j) \in \mathcal{E}$. 
\end{itemize}

\section{Distributed Attack Resilient Frequency control}\label{Frequency_control}
For resiliency towards frequency attacks $d_{\omega}$, we propose the system $S_{\omega}$ in \eqref{coupled_model} as follows:
\begin{subequations}\label{u_w}
	\begin{align}
		\dot{\omega}  &= A\omega + \beta A z + B\omega^{*} + Ld_{\omega} \label{u_w_1}\\
		\dot{z} &= A z - \beta A \omega + \beta C \omega^{*}, \label{u_w_2}
	\end{align}
\end{subequations}
where $\omega^{*}$ is the reference frequency fed to the leader DGs in the group. The matrix $A$ is given by
\begin{align}\label{A}
	A = -\begin{bmatrix}
		\sum_{j =1}^{n} a_{1j} & -a_{12} & \ldots & -a_{1n}\\
		-a_{21} & \sum_{j =1}^{n} a_{2j} & \ldots & -a_{2n} \\
		\vdots & \vdots & \vdots & \vdots \\
		-a_{n1} & -a_{n2} & \ldots & \sum_{j =1}^{n} a_{nj}
	\end{bmatrix} \nonumber \\
	- \begin{bmatrix}
		g_1 & 0 & \ldots & 0\\
		0 & g_2 & \ldots & 0 \\
		\vdots & \vdots & \vdots & \vdots \\
		0 & 0 & \ldots & g_n
	\end{bmatrix}
	= -(L+G),
\end{align}
where $G$ is a diagonal matrix with its entries $g_i$ represents the pinning gain, such that $g_i = 1$ if the $i^\text{th}$ DG is a leader and $0$ otherwise. Further, $\beta > 0$ is a gain factor and the vectors $B$ and $C$ are defined as $B \coloneqq G\pmb{1}_n$ and $C \coloneqq A\pmb{1}_n$. In \eqref{u_w}, matrices $A$ and $B$ are such that the frequency controller $\dot{\omega}_i$ of the $i^\text{th}$ DG contains relative frequency term $\sum_{j \in \mathcal{N}_i}(\omega_j - \omega_i)$ (because of the Laplacian $L$) responsible for frequency consensus, and the term $g_i(\omega_i - \omega^\star)$ associated with the pinning gain $g_i$, shapes the solution trajectories such that consensus occurs at the desired frequency $\omega^\star$. The interconnection between the two layers $\Sigma$ and $\Pi$ is realized by the matrix $\beta A$, associated with (i) the auxiliary state vector $z$ in the $\dot{\omega}$ in dynamics \eqref{u_w_1}, and (ii) the actual frequency vector $\omega$ in the $\dot{z}$ dynamics \eqref{u_w_2}. Note that the auxiliary state vector $z$ has no physical significance and may assume any steady-sate value. We have the following lemma relating matrices $A, B, C$ and $G$: 

\begin{lem}\label{lem_matrices}
Consider system \eqref{u_w} with matrix $A$ in \eqref{A}. Then, $A\pmb{1}_n = -B$ and $C = -G$ for an undirected and connected graph $\mathcal{G}$.
\end{lem}
\begin{proof}
Multiplying by $\pmb{1}_n$ on both sides of \eqref{A}, we have that $A \pmb{1}_n = -L\pmb{1}_n - G \pmb{1}_n = - G \pmb{1}_n = - B$, since $L \pmb{1}_n = \pmb{0}_n$ for an undirected and connected graph. Further, since $C \coloneqq A\pmb{1}_n$ by definition, it follows that $C = -G$. 
\end{proof}

For further analysis, let us introduce the frequency error vector
\begin{equation} \label{error_transformation}
e_{\omega} = \omega - \omega^{*}\pmb{1}_n,	
\end{equation}
using which, \eqref{u_w} can be expressed as:
\begin{align*}
	\dot{e}_{\omega} &= A e_{\omega} + \beta A z + A\omega^{*}\pmb{1}_n + B \omega^{*} + L d_{\omega} \\
	\dot{z} &= Az - \beta A (e_{\omega} + \omega^{*}\pmb{1}_n) + \beta C \omega^{*}.
\end{align*}
Following Lemma~\ref{lem_matrices}, the above equations are simplified as:
\begin{subequations}\label{w_and_z_dot}
	\begin{align} 
		\dot{e}_{\omega} &= A e_{\omega} + \beta A z + L d_{\omega}\\
		\dot{z} &= Az - \beta A e_{\omega}.
	\end{align}
\end{subequations}
Let $\xi_{\omega} = [e_{\omega}^T, z^T]^T \in \mathbb{R}^{2n}$ be the joint state vector, using which, \eqref{w_and_z_dot} can be written in the compact form as:
\begin{equation}\label{xi_dot}
	\dot{\xi}_{\omega} = \mathcal{K} \xi_{\omega} + D_{\omega},
\end{equation}
where $D_{\omega} = [(Ld_{\omega})^T, \pmb{0}_n^T]^T \in \mathbb{R}^{2n}$ and the block matrix $\mathcal{K}_{2n \times 2n}$ is given by
\begin{equation}    \label{K}
	\mathcal{K} = \begin{bmatrix}
		A & \beta A \\
		- \beta A & A
	\end{bmatrix}.
\end{equation}

\begin{lem}\label{lem_K_Hurwitz}
The block matrix $\mathcal{K}$ in \eqref{K} is Hurwitz.  
\end{lem}

\begin{proof}
	Using Kronecker product, \eqref{K} can be written as $\mathcal{K} = \Phi \otimes A$, where 
	\begin{equation}\label{phi}
		\Phi = \begin{bmatrix}
			1 & \beta  \\
			- \beta  & 1
		\end{bmatrix},
	\end{equation}
	is a $2 \times 2$ matrix with eigenvalues $\lambda(\Phi) = 1 \pm j_c \beta$. Further, since graph $\mathcal{G}$ is undirected and connected, $L$ is positive semi-definite \cite{wang2010}. Also, matrix $G$ contains at least one non-zero diagonal entry (as there is at least one leader in the group). Using Lemma~\ref{lem_1} from Subsection~\ref{Preliminaries}, it can be stated that $L+G$ is symmetric and positive-definite. Alternatively, $A = -(L+G)$, as defined in \eqref{A}, is symmetric and negative-definite, and hence, has all real and strictly negative eigenvalues $\lambda_i(A) < 0, \forall i = 1, \ldots, n$. Now, using the multiplicative property of Kronecker product \cite[Theorem~4.2.12]{horn1991}, it can be inferred that $\lambda_i(\mathcal{K}) = \lambda_i(\Phi) \lambda_i(A), \forall i = 1, \ldots, 2n$. Alternatively, all $2n$ eigenvalues of $\mathcal{K}$ are complex conjugate having strictly negative real part, implying that $\mathcal{K}$ is Hurwitz.   
\end{proof}

For better motivation and simplicity, we first discuss that the proposed framework \eqref{xi_dot} assures the frequency consensus in the absence of an attack, followed by the result in the attacked scenario.

\begin{lem}[Frequency control in absence of attack]\label{lem_freq_attack_absence}
If $d_{\omega} \equiv \pmb{0}_n$, the dynamics \eqref{xi_dot} assures that $\omega \to \omega^\star \pmb{1}_n$ as $t \to \infty$. 
\end{lem}

\begin{proof}
For $d_{\omega} \equiv \pmb{0}_n$, \eqref{xi_dot} becomes $\dot{\xi}_{\omega} = \mathcal{K} \xi_{\omega}$ and has the solution $\xi_{\omega} (t) = {\rm e}^{\mathcal{K} t} \xi_{\omega} (0)$. Now, it is straightforward to conclude that  $\xi_{\omega} \to \pmb{0}_{2n}$ at $t \to \infty$, since $\mathcal{K}$ is Hurwitz (see Lemma~\ref{lem_K_Hurwitz}). In other words, $\omega \to \omega^\star \pmb{1}_n$ as $t \to \infty$ in absence of $d_{\omega}$ (the steady-state value of the auxiliary state $z$ is not important).
\end{proof}

If $d_{\omega} \neq \pmb{0}_n$, we have the following convergence theorem for microgrid frequency, addressing the problem {\bf (P1)}. 
\begin{thm}[Frequency control in presence of attack]\label{thm_1}
Consider the system \eqref{xi_dot} where the frequency attack signal $d_{\omega} \neq \pmb{0}_n$ satisfies Assumption~\ref{assumption}. Then, for a sufficiently large value of gain $\beta$, the frequencies of all DGs remain in a small neighborhood of nominal frequency, i.e., $\|\omega - \omega^\star \pmb{1}_n\| \leq \epsilon_{\omega}$ for some small $\epsilon_{\omega} > 0$.
\end{thm}

Before proceeding to the proof, we first discuss the below important result: 

\begin{lem}\label{lem_H_beta}
Let $\mathcal{H}_{\beta} = \begin{bmatrix}
	(A + {\beta}^{2}A)^{-1} \\
	\beta (A + {\beta}^2 A)^{-1} 
\end{bmatrix}$ be a block matrix of dimension $2n \times n$, where $A$ is given by \eqref{A}. The operator norm of $\mathcal{H}_{\beta}$ is given by
\begin{equation*}
\|\mathcal{H}_{\beta}\| =  \frac{1}{\lambda_{\min}(L + G) \sqrt{1 + \beta^2}}.
\end{equation*}  
\end{lem}

\begin{proof}
Using Kronecker product, the matrix $\mathcal{H}_{\beta}$ can be rewritten as $\mathcal{H}_{\beta} = \begin{bmatrix}
	\frac{1}{1 + \beta^2} \\
	\frac{\beta}{1 + \beta^2} 
\end{bmatrix} \otimes A^{-1}$. Taking operator norm on both sides, we have 
\begin{equation*}
\|\mathcal{H}_{\beta}\| = \left\| \begin{bmatrix}
	\frac{1}{1 + \beta^2} \\
	\frac{\beta}{1 + \beta^2} 
\end{bmatrix} \otimes A^{-1} \right\| = \left\| \begin{bmatrix}
\frac{1}{1 + \beta^2} \\
\frac{\beta}{1 + \beta^2} 
\end{bmatrix} \right\| \|A^{-1}\|,
\end{equation*}
using the property $\|\mathcal{A} \otimes \mathcal{B}\| = \|\mathcal{A}\|\|\mathcal{B}\|$ of the operator norms for Kronecker products of two matrices $\mathcal{A}$ and $\mathcal{B}$ \cite[Theorem 8, p. 412]{lancaster1972}. Note that $\left\| \begin{bmatrix}
	\frac{1}{1 + \beta^2} \\
	\frac{\beta}{1 + \beta^2} 
\end{bmatrix} \right\| = \frac{1}{\sqrt{1 + \beta^2}}$ and $\|A^{-1}\| = \sqrt{\lambda_{\max}((A^{-1})^T A^{-1})}$. From \eqref{A}, it is straightforward to check that $A^T = A \implies (A^T)^{-1} = A^{-1} \implies (A^{-1})^T = A^{-1}$. The last relation can also be rewritten as $((-A)^{-1})^T = (-A)^{-1}$, where matrix $-A$ is positive-definite. Using this, one can write $\|A^{-1}\| = \sqrt{\lambda_{\max}((-A^{-1})^T (-A)^{-1})} = \lambda_{\max}(-A^{-1}) = \frac{1}{\lambda_{\min}(-A)}$. Consequently, it can be concluded using \eqref{A} that $\|\mathcal{H}_{\beta}\| =  \frac{1}{\lambda_{\min}(L + G) \sqrt{1 + \beta^2}}$, and hence, proving the result.  
\end{proof}

We are now ready to prove Theorem~\ref{thm_1}.

\begin{proof}[Proof of Theorem~\ref{thm_1}]
If $d_{\omega} \neq \pmb{0}_n$, the solution of linear system \eqref{xi_dot} is obtained as $\xi_{\omega} (t) = {\rm e}^{\mathcal{K} t} \xi_{\omega} (0) + \int_{0}^{t} {\rm e}^{\mathcal{K}(t- \tau)}D_{\omega} d\tau$. Now, taking $\mathcal{L}_2$ norm on both sides and applying the Cauchy–Schwarz inequality results in $\lim_{t \to \infty} \|\xi_{\omega}(t) \| \leq \lim_{t \to \infty} \|{\rm e}^{\mathcal{K}t} \xi_{\omega}(0)\| + \lim_{t \to \infty} \|\int_{0}^{t} {\rm e}^{\mathcal{K}(t- \tau)}D_{\omega} d\tau \|$. Since $\lim_{t \to \infty} \|{\rm e}^{\mathcal{K}t} \xi_{\omega}(0) \| \to 0$, as $\mathcal{K}$ is Hurwitz (see Lemma~\ref{lem_K_Hurwitz}), it can be written that $\lim_{t \to \infty} \|\xi_{\omega}(t)\| \leq \lim_{t \to \infty} \|\int_{0}^{t} {\rm e}^{\mathcal{K}(t- \tau)}D_{\omega} d\tau \|$. Note that $D_{\omega}$ is uniformly bounded, as it is $d_{\omega}$, according to Assumption~\ref{assumption}. Now, it immediately follows from Lemma~\ref{lem_2} (from Subsection~\ref{Preliminaries}) that there exists a time-independent vector $\hat{d}_{\omega} \in \mathbb{R}^n$ with $\|\hat{d}_{\omega}\| \leq \bar{D}_{\omega}$ (where $\bar{D}_{\omega}$ is defined in Assumption~\ref{assumption}) such that $\|\int_{0}^{t} {\rm e}^{\mathcal{K}(t- \tau)}D_{\omega}(\tau) d\tau \| \leq \|\int_{0}^{t} {\rm e}^{\mathcal{K}(t- \tau)}\hat{D}_{\omega} d\tau \|$ for some constant vector $\hat{D}_{\omega} = [(L \hat{d}_{\omega})^T, \pmb{0}_n^T]^T \in \mathbb{R}^{2n}$ for all $t \geq T$. This implies that $\lim_{t \to \infty} \|\xi_{\omega}(t) \| \leq \lim_{t \to \infty} \|\int_{0}^{t} {\rm e}^{\mathcal{K}(t- \tau)}\hat{D}_{\omega} d\tau \|$, which on further simplification results in $\lim_{t \to \infty} \| \xi_{\omega}(t) \|\leq \| \mathcal{K}^{-1} \hat{D}_{\omega} \|$. Note that $\mathcal{K}$ is invertible as it has all non-zero eigenvalues (Lemma~\ref{lem_K_Hurwitz}) and its inverse is given by \cite[Theorem~0.7.3]{horn2012matrix}:
\begin{equation}\label{K_inverse}
\mathcal{K}^{-1} = \begin{bmatrix}
         \mathcal{M}_{11} & \mathcal{M}_{12}  \\
        \mathcal{M}_{21}  & \mathcal{M}_{22}
    \end{bmatrix},
\end{equation}
using which, it holds that
\begin{equation}\label{e_w_z_norm}
\lim_{t \to \infty}  \| \xi_{\omega}(t) \| \leq \left\| 
\begin{bmatrix}
		 \mathcal{M}_{11} (L\hat{d}_{\omega}) \\
		\mathcal{M}_{21} (L\hat{d}_{\omega}) 
	\end{bmatrix} \right\|,
\end{equation}
where $\mathcal{M}_{11} = (\mathcal{K}_{11} - \mathcal{K}_{12} \mathcal{K}_{22}^{-1} \mathcal{K}_{21})^{-1}$ and $\mathcal{M}_{21} = \mathcal{K}_{22}^{-1} \mathcal{K}_{21}(\mathcal{K}_{12} \mathcal{K}_{22}^{-1} \mathcal{K}_{21} - \mathcal{K}_{11})^{-1}$ \cite[Theorem~0.7.3]{horn2012matrix}, with $\mathcal{K}_{11} = \mathcal{K}_{22} = A$ and $\mathcal{K}_{12} = -\mathcal{K}_{21} = \beta A$, from \eqref{K}. Substituting these, we get $\mathcal{M}_{11} = (A + {\beta}^{2}A)^{-1}$ and $\mathcal{M}_{21} =  \beta (A + {\beta}^2 A)^{-1}$, and hence, 
\begin{equation*} 
\lim_{t \to \infty} \left\| \begin{bmatrix}
        e_{\omega} (t) \\
        z (t)
     \end{bmatrix} \right\| \leq \left\| \begin{bmatrix}
         (A + {\beta}^{2}A)^{-1} \\
         \beta (A + {\beta}^2 A)^{-1} 
     \end{bmatrix} \right\| \|L\hat{d}_{\omega}\| = \|\mathcal{H}_{\beta}\| \|L\hat{d}_{\omega}\|,
\end{equation*}
where we have replaced the first term by $\|\mathcal{H}_{\beta}\|$, using Lemma~\ref{lem_H_beta}. Since $\|L\hat{d}_{\omega}\| \leq \|L\| \|\hat{d}_{\omega}\| = \sqrt{\lambda_{\max}(L^T L)} \|\hat{d}_{\omega}\| = \lambda_{\max}(L) \|\hat{d}_{\omega}\|$ (since $L = L^T$ for an undirected and connected graph $\mathcal{G}$) and $\|\hat{d}_{\omega}\| \leq \bar{D}_{\omega}$ (see Assumption~\ref{assumption}), it holds that 
\begin{equation} \label{e_w_norm}
\lim_{t \to \infty} \left\| \begin{bmatrix}
	e_{\omega} (t) \\
	z (t)
\end{bmatrix} \right\| \leq \frac{\lambda_{\max}(L) \bar{D}_{\omega}}{\lambda_{\min}(L + G) \sqrt{1 + \beta^2}},
\end{equation}
exploiting Lemma~\ref{lem_H_beta}. From \eqref{e_w_norm}, one can conclude that $\|e_{\omega}(t)\| \leq \epsilon_{\omega}$ as $t \to \infty$, with $\epsilon_{\omega} = \frac{\lambda_{\max}(L) \bar{D}_{\omega}}{\lambda_{\min}(L + G) \sqrt{1 + \beta^2}}$. Further, since $\epsilon_{\omega}$ is small for sufficiently large values of gain $\beta$ (which is a design parameter), $\omega$ converges to a small neighborhood around ${\omega}^{*}\pmb{1}_n$ as $t \to \infty$, using \eqref{error_transformation}. 
\end{proof}

\begin{remark}
It can be inferred from inequality \eqref{e_w_norm} that the steady-state bound $\epsilon_{\omega}$ is:
\begin{itemize}
\item proportional to the maximum eigenvalue $\lambda_{\max}(L)$ of Laplacian $L$ and the attack bound $\bar{D}_{\omega}$. If the graph $\mathcal{G}$ has minimal connectivity (i.e., small $\lambda_{\max}(L)$), the attacker will have fewer links to target, and hence, will have relatively less impact.
\item inversely proportional to the minimum eigenvalue $\lambda_{\min}(L+G)$ of matrix $L+G$ and the design parameter $\beta$. Beside selecting large $\beta$, it is possible to make $\lambda_{\min}(L+G)$ large by choosing the large value of pinning gain $g_i$, since $\lambda_{\min}(L+G) \geq \lambda_{\min}(L) + \min_i \{g_i\} = \min_i \{g_i\}$, using Weyl's theorem \cite[Chapter 4, p. 239]{horn2012matrix} and the fact that $\lambda_{\min}(L) = 0$. Therefore, $\lambda_{\min}(L+G)$ primarily influenced by the pinning gain $g_i$. 
\end{itemize} 
In summary, it is evident that the effect of varying $\beta$ is more dominant on the value of $\epsilon_{\omega}$, as compared to $g_i$, and can be decided appropriately by the user. 
\end{remark}

\section{Distributed Attack Resilient Power Control}\label{Power_control}
Unlike frequency control, the implementation of the distributed power controller is done in a leaderless configuration to assure the proportional power-sharing among the DGs, leading to the system $S_{P}$ in \eqref{coupled_model} as follows:
\begin{subequations}\label{u_P}
\begin{align}
m_P\dot{P} &= -L(m_P P) + \beta L z + L{d_P} \label{u_P_1}\\
\dot{z} &= -L z - \beta L(m_P P), \label{u_P_2}
\end{align}
\end{subequations}
where Laplacian $L$ is associated with the state vector $m_P P$, instead of matrix $A$ as in \eqref{u_w}. Here, the interconnection between $\Sigma$ and $\Pi$ is realized by the matrix $\beta L$. Defining power-sharing error as (where $\Delta_P$ is defined in Problem {\bf (P2)})
\begin{equation}\label{power_error}
e_P = m_P P - \Delta_P\pmb{1}_n,	
\end{equation}
\eqref{u_P} can be expressed in terms of $e_P$ as:
\begin{subequations}\label{power_error_dynamics}
\begin{align}
\dot{e}_P &= -L e_P + \beta L z + L d_P\\
\dot{z} &= - \beta L e_P - L z. 
\end{align}
\end{subequations}

We have the following lemma in absence of power attacks: 

\begin{lem}[Power-sharing control in absence of attack]\label{lem_power_attack_absence}
If $d_{P} \equiv \pmb{0}_n$, the dynamics \eqref{power_error_dynamics} assures that $P \to \Delta_P \pmb{1}_n$ as $t \to \infty$. 
\end{lem}

\begin{proof}
Consider the candidate Lyapunov function $V_P = 0.5e_p^T e_p + 0.5z^Tz$, whose time-derivative along \eqref{power_error_dynamics} with $d_{P} \equiv \pmb{0}_n$ is obtained as $\dot{V}_p = e_p^T \dot{e}_p + z^T \dot{z} = e_p^T[-L(m_P P) + \beta Lz] + z^T [-Lz - \beta L (m_p P)]$, where we used the fact that $L \Delta_P \pmb{1}_n = \pmb{0}_n$ for an undirected and connected graph $\mathcal{G}$. On further simplification, one can get that $\dot{V}_P = -e_p^T L e_p - z^T L z \leq 0$, which is negative semi-definite. Now, using LaSalle's invariance theorem \cite[Corollary 4.2, p. 129]{khalil2002nonlinear}, it can be concluded that $e_P = z = \pmb{0}_n$ is the desired equilibrium point, implying that, $P \to \Delta_P \pmb{1}_n$ as $t \to \infty$, using \eqref{power_error}.   
\end{proof}

In case $d_{P} \neq \pmb{0}_n$, the analysis is different from the earlier Theorem~\ref{thm_1}, as the system matrix $L$ in \eqref{power_error_dynamics} has one zero eigenvalue and rest positive eigenvalues for an undirected and connected graph $\mathcal{G}$. To proceed further, we filter out the dynamics associated with the simple zero eigenvalue by using the transformation:	
\begin{equation}\label{T}
	e_P = T\begin{bmatrix}
		\bar{e}_P   \\
		\tilde{e}_P 
	\end{bmatrix}, 
	\     
	z = T\begin{bmatrix}
		\bar{z}   \\
		\tilde{z} 
	\end{bmatrix},
	\
	d_P = T\begin{bmatrix}
		\bar{d}_P   \\
		\tilde{d}_P 
	\end{bmatrix},
\end{equation}
where $\bar{e}_P, \bar{z}, \bar{d}_P \in \mathbb{R}$ and $\tilde{e}_P, \tilde{z}, \tilde{d}_P \in \mathbb{R}^{n-1}$ and the transformation $T = [v_1,\ldots, v_n]$ with $v_i$ being the left eigenvector associated with $\lambda_i(L)$, $i = 1, \ldots, n$. Using transformation $T$, $L$ can be diagonalized as:
\begin{equation}\label{V}
	T^{-1} L T = \begin{bmatrix}
		\lambda_1(L)  & 0 & \cdots & 0  \\
		0  & \lambda_2(L) & \cdots & 0  \\
		\vdots  & \vdots & \ddots & \vdots   \\
		0  & 0 & \cdots & \lambda_n(L)
	\end{bmatrix}  
	= 
	\begin{bmatrix}
		\begin{array}{c|c}
			0 & \pmb{0}^T_{n - 1} \\
			\hline
			\pmb{0}_{n-1} & \mathcal{R}
		\end{array}
	\end{bmatrix},
\end{equation}
where $\mathcal{R} \in \mathbb{R}^{(n-1)\times (n-1)}$ is a diagonal matrix comprising non-zero eigenvalues of $L$. Using \eqref{T}, \eqref{power_error_dynamics} can be rewritten as:
\begin{subequations}\label{transformed_power_error_dynamics}
	\begin{align}    
		\begin{bmatrix}
			\dot{\bar{e}}_P   \\
			\dot{\tilde{e}}_P 
		\end{bmatrix} &= 
		-T^{-1} L T\begin{bmatrix}
			\bar{e}_P   \\
			\tilde{e}_P 
		\end{bmatrix}
		+
		\beta T^{-1} L T\begin{bmatrix}
			\bar{z}   \\
			\tilde{z} 
		\end{bmatrix}
		+
		T^{-1} L T\begin{bmatrix}
			\bar{d}_P   \\
			\tilde{d}_P 
		\end{bmatrix},\\
		\begin{bmatrix}
			\dot{\bar{z}}   \\
			\dot{\tilde{z}}
		\end{bmatrix} &= 
		- \beta T^{-1} L T\begin{bmatrix}
			\bar{e}_P   \\
			\tilde{e}_P 
		\end{bmatrix}
		-
		T^{-1} L T\begin{bmatrix}
			\bar{z}   \\
			\tilde{z} 
		\end{bmatrix}.
	\end{align}
\end{subequations}
Now, using \eqref{V}, since $\dot{\bar{e}}_P = \dot{\bar{z}} = 0$, \eqref{transformed_power_error_dynamics} further reduces to
\begin{subequations}
	\begin{align}
		\dot{\tilde{e}}_p &= - \mathcal{R} \tilde{e}_P + \beta \mathcal{R} \tilde{z} + \mathcal{R} \tilde{d}_p \\
		\dot{\tilde{z}} &= - \beta \mathcal{R} \tilde{e}_P - \mathcal{R} \tilde{z},
	\end{align}
	which can be expressed in the compact form as:
\end{subequations}
\begin{equation} \label{xi_p_dot_d}
	\dot{\xi}_P = \tilde{\mathcal{K}}\xi_P + \tilde{D}_P,
\end{equation}
where $\xi_P = [{\tilde{e}_P}^T, \tilde{z}^T]^T \in \mathbb{R}^{2(n-1)}$, $\tilde{D}_P = [(\mathcal{R} \tilde{d}_p)^T, \pmb{0}_{n-1}^T]^T \in \mathbb{R}^{2(n-1)}$ and 
\begin{align}\label{K_p}
	\tilde{\mathcal{K}} = 
	\begin{bmatrix}
		-\mathcal{R} & \beta \mathcal{R} \\
		- \beta \mathcal{R} & -\mathcal{R}
	\end{bmatrix}
	= \begin{bmatrix}
		-1 & \beta  \\
		- \beta  & -1
	\end{bmatrix}
	\otimes \mathcal{R} 
	= \tilde{\Phi} \otimes \mathcal{R}. 
\end{align}
Since $\lambda(\tilde{\Phi}) = -1 \pm j_c \beta$ and $\mathcal{R}$ is a positive definite matrix by construction, $\tilde{\mathcal{K}}$ is Hurwitz and has $2(n-1)$ strictly negative eigenvalues in accordance with Lemma~\ref{lem_K_Hurwitz}. We are now ready to state the following convergence result, solving Problem {\bf (P2)}: 

\begin{thm}[Power-sharing control in presence of attack]\label{thm_2}
Consider the system \eqref{power_error_dynamics} where the power attack signal $d_{P} \neq \pmb{0}_n$ satisfies Assumption~\ref{assumption}. Then, for a sufficiently large value of $\beta$, the vector $m_P P$ remains in a small neighborhood of its nominal value while sharing proportional active power, i.e., $\|m_P P - \Delta_P \pmb{1}_n \| \leq \epsilon_P$ for some small $\epsilon_P$.  
\end{thm}

\begin{proof}
The solution of \eqref{xi_p_dot_d} is given by $\xi_P (t) = {\rm e}^{\tilde{\mathcal{K}} t} \xi_P (0) + \int_{0}^{t} {\rm e}^{\tilde{\mathcal{K}}(t- \tau)}\tilde{D}_p d\tau$. Upon taking norm on both sides and following similar steps as in the proof of Theorem~\ref{thm_1}, one can obtain $\lim_{t \to \infty} \|\xi_P(t) \|\leq \| \tilde{\mathcal{K}}^{-1} \hat{D}_P \|$, where $\hat{D}_P = [(\mathcal{R}\hat{d}_P)^T, \pmb{0}_{n-1}^T]^T \in \mathbb{R}^{2(n-1)}$ for some time-independent vector $\hat{d}_P \in \mathbb{R}^{n-1}$ satisfying $\|\hat{d}_P\| \leq \bar{D}_P$ (see Assumption~\ref{assumption} for $\bar{D}_P$) such that $\|\int_{0}^{t} {\rm e}^{\tilde{\mathcal{K}}(t- \tau)}\tilde{D}_p(\tau) d\tau \| \leq \|\int_{0}^{t} {\rm e}^{\tilde{\mathcal{K}}(t- \tau)}\hat{D}_P d\tau \|$, according to Lemma~\ref{lem_2} from Subsection~\ref{Preliminaries}. Now, similar to \eqref{e_w_norm}, it can be concluded that  
\begin{equation} 
 \lim_{t \to \infty} \left\| \begin{bmatrix}
    	\tilde{e}_P (t) \\
    	\tilde{z} (t)
    \end{bmatrix} \right\| \leq \lambda_{\max}(\mathcal{R})\left\| \begin{bmatrix}
         (\mathcal{R} + {\beta}^{2}\mathcal{R})^{-1} \\
         \beta (\mathcal{R} + {\beta}^2 \mathcal{R})^{-1} 
     \end{bmatrix} \right\| \bar{D}_P.
\end{equation}
Now, exploiting Lemma~\ref{lem_H_beta} by replacing $A$ by $\mathcal{R}$, it can be obtained that
\begin{equation} \label{e_p_norm}
	\lim_{t \to \infty} \left\| \begin{bmatrix}
		\tilde{e}_P (t) \\
		\tilde{z} (t)
	\end{bmatrix} \right\| \leq \frac{\lambda_{\max}(\mathcal{R})\bar{D}_P}{\lambda_{\min}(\mathcal{R}) \sqrt{1 + \beta^2}} = \frac{\lambda_{\max}(L)\bar{D}_P}{\lambda_2(L) \sqrt{1 + \beta^2}},
\end{equation}
since $\lambda_{\max}(\mathcal{R}) = \lambda_{\max}(L)$ and $\lambda_{\min}(\mathcal{R}) = \lambda_2(L)$ by construction, where $\lambda_2(L)$ is the Fielder eigenvalue of the Laplacian $L$. From \eqref{e_p_norm}, it can be concluded that $\|\tilde{e}_{P}(t)\| \leq \epsilon_{P}$ (and hence, $\|{e}_{P}(t)\| \leq \epsilon_{P}$) as $t \to \infty$, where $\epsilon_P = \frac{\lambda_{\max}(L)\bar{D}_P}{\lambda_2(L) \sqrt{1 + \beta^2}}$. Furthermore, for sufficiently large values of gain $\beta$, $\epsilon_P$ assumes small value, and hence, $m_P P$ converges to a small neighborhood around $\Delta_P\pmb{1}_n$ as $t \to \infty$, using \eqref{power_error}. 
\end{proof}

\section{Attack Detection: Exploring Interaction Between Control and Auxiliary Layers}\label{Attack_detection}
Addressing Problem {\bf (P3)}, this section proposes an attack detection method relying on the interaction between the control layer $\Sigma$ and auxiliary layer $\Pi$ in Fig.~\ref{control_digram}. For brevity and clarity, we provide a discussion for the frequency attacks $d_{\omega}$ affecting dynamics $S_{\omega}$ in \eqref{coupled_model}. However, it is equally applicable to the power attacks $d_P$ appearing in the dynamics $S_P$. 

\begin{figure}[h]
	\centering
	\includegraphics[scale=0.15]{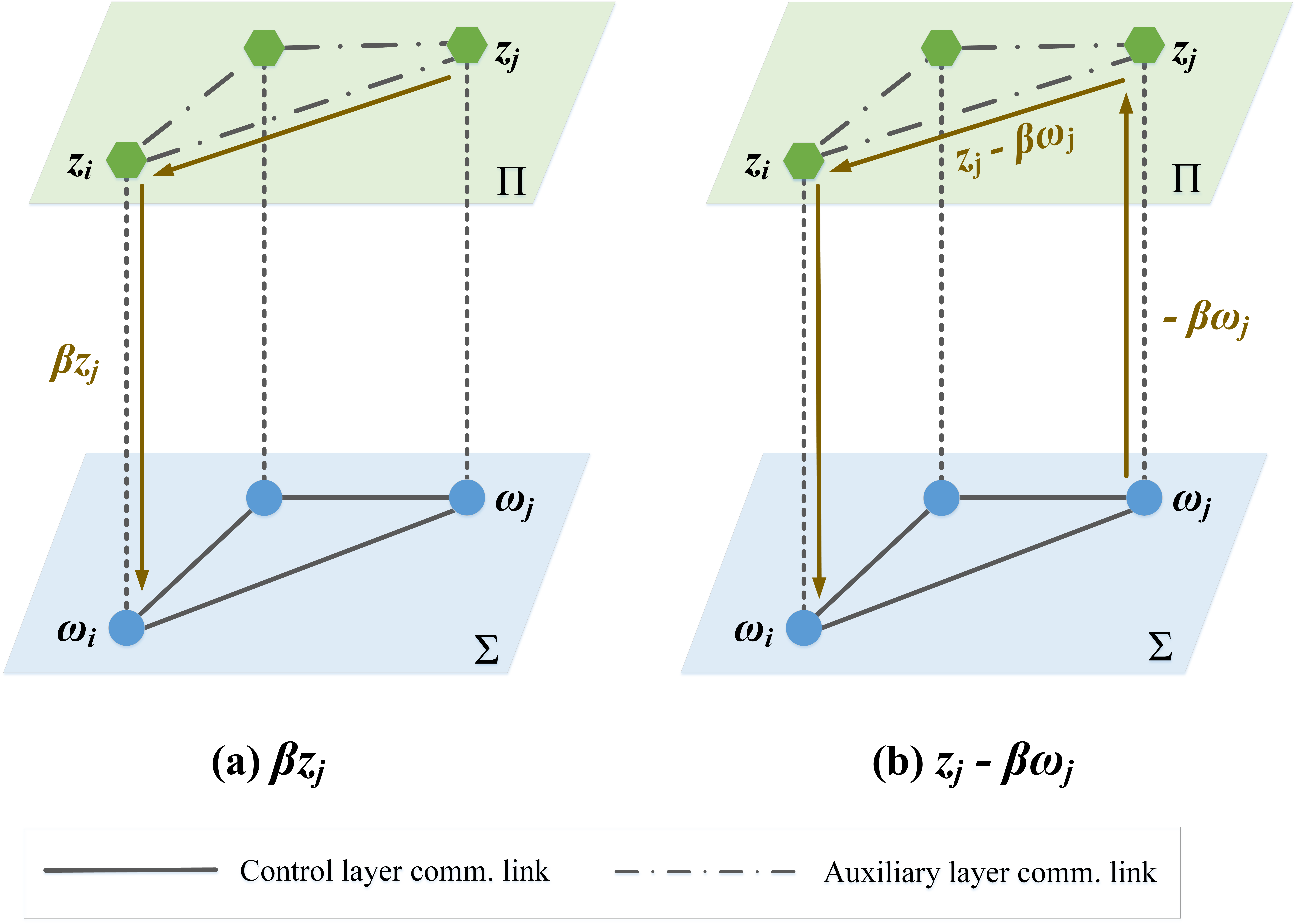}
	\caption{Information flow due to interaction network ($\beta A$) between control ($\Sigma$) and auxiliary ($\Pi$) layers.}
	\label{ij_bar}
	\vspace*{-10 pt}
\end{figure}

Since the layer $\Sigma$ is vulnerable to attacks injected by the attacker as compared to the secured and hidden layer $\Pi$, one can exploit the layer $\Pi$ to check the authenticity of the information shared between two nodes in the layer $\Sigma$. As shown in Fig.~\ref{control_digram}, the layer $\Pi$ contains a set of (virtual) nodes, directly associated with individual (actual) nodes in the layer $\Sigma$. Consequently, the $i^{\text{th}}$ node in $\Sigma$ has state information $z_i$ of the associated node in $\Pi$ and similarly the $i^{\text{th}}$ node in $\Pi$ has information $\omega_i$ of the associated node in $\Sigma$. Further, both the layers are connected to each other via interconnection matrices $\beta A$. With respect to the $i^{\text{th}}$ node, \eqref{u_w} can be rewritten, by segregating terms on the basis of the contributions by the $i^{\text{th}}$ node itself and the neighboring nodes in both $\Sigma$ and $\Pi$ layers, as follows \cite{gusrialdi2022cooperative}: 
\begin{subequations}\label{attack_detection}
\begin{align}
\label{i_1} \dot{\omega}_i  &= A_{[i]} \omega  - \beta (|\mathcal{N}_i | + g_i ) z_i + \sum_{j \in \mathcal{N}_i} \beta z_j + g_i {\omega}^{*} + L_{[i]} d_{\omega}\\
\nonumber \dot{z}_i &= - (| \mathcal{N}_i | + g_i) z_i + \beta (| \mathcal{N}_i | + g_i ) \omega_{i} + \sum_{j \in  \mathcal{N}_i} (z_j - \beta {\omega}_j) \\
\label{i_2} & \qquad - \beta g_i {\omega}^{*},
\end{align}
\end{subequations}
where $A_{[i]}, L_{[i]}$ denote the $i^{\text{th}}$ row of matrices $A$ and $L$, respectively. Further, $g_i = 1$ if the node $i$ is the leader, else $g_i = 0$. The following observations are straightforward toward practical implementation of \eqref{attack_detection}: 
\begin{itemize}
	\item In \eqref{i_1}, the $i^{\text{th}}$ node in $\Sigma$ has access to the information $\beta z_{j}$ of the neighboring nodes in $\Pi$ via its corresponding $i^{\text{th}}$ node in $\Pi$, shared through the interaction network between $\Sigma$ and $\Pi$, as shown in Fig.~\ref{ij_bar}(a). The implementation of remaining terms in \eqref{i_1} is straightforward. 
	\item Analogously, in \eqref{i_2}, the $i^{\text{th}}$ node in $\Pi$ has information $z_j - \beta \omega_j$ from the neighboring nodes in $\Pi$ itself where the information $-\beta \omega_j$ is shared through the interaction between the two layers, as shown Fig. \ref{ij_bar}(b). Again, the implementation of the remaining terms in \eqref{i_2} is straightforward. 
\end{itemize}
From the above discussion, it is clear that the $i^\text{th}$ node in $\Sigma$ has access to the following two pieces of additional information due to the interaction between $\Sigma$ and $\Pi$:
\begin{align}\label{detection}
	\bar{z}_{ij} = \beta z_j, \qquad \bar{\omega}_{ij} = z_j - \beta {\omega}_{j},
\end{align}
where notations $\bar{\omega}_{ij}$ and $\bar{z}_{ij}$ are used to emphasize that these signals contain information about the neighboring nodes $j$ in $\Pi$. According to \eqref{detection}, the $i^\text{th}$ node  estimates the frequency received from the $j^\text{th}$ node in $\Pi$ as:
\begin{equation}\label{test}
	\hat{\omega}_{ij} = \frac{1}{\beta} \left[\frac{\bar{z}_{ij}}{\beta} - \bar{\omega}_{ij} \right],
\end{equation}
which can be compared with the actual frequency signal $\check{\omega}_{ij}$ in \eqref{after_attack} received at the $i^\text{th}$ node to testify whether the communication link $(i, j) \in \mathcal{E}$ in $\Sigma$ is under attack or not. Consequently, if $\hat{\omega}_{ij} \neq \check{\omega}_{ij} $, the corresponding $(i, j)^\text{th}$ link is declared to be under attack. Once the corrupted communication link is identified, it can be isolated (provided the remaining network contains a spanning tree) from the rest of the control layer till the time attack remains prevalent in the network. 

\begin{remark}
It is essential to emphasize that the attack detection approach described above is a byproduct of the interaction between $\Sigma$ and $\Pi$. It can be employed if there is a specific need to isolate any corrupted link, thereby enhancing the accuracy of the steady-state behavior. However, it's crucial to note that the isolation of any link is not mandatory for the proper functioning of the proposed resilient controllers \eqref{u_w} and \eqref{u_P} for some finite $\epsilon_{\omega}$ and $\epsilon_P$. Moreover, the isolation between the layers $\Sigma$ and $\Pi$ can be ensured using contemporary communication network slicing approaches, such as cloud computing management and software-defined networking \cite{Danzi2019, Zhang2019}. These advanced techniques provide effective mechanisms for segregating communication channels and ensuring the secure operation of the system.	
\end{remark}

\begin{figure}[t]
	\centering
	\includegraphics[scale=0.55]{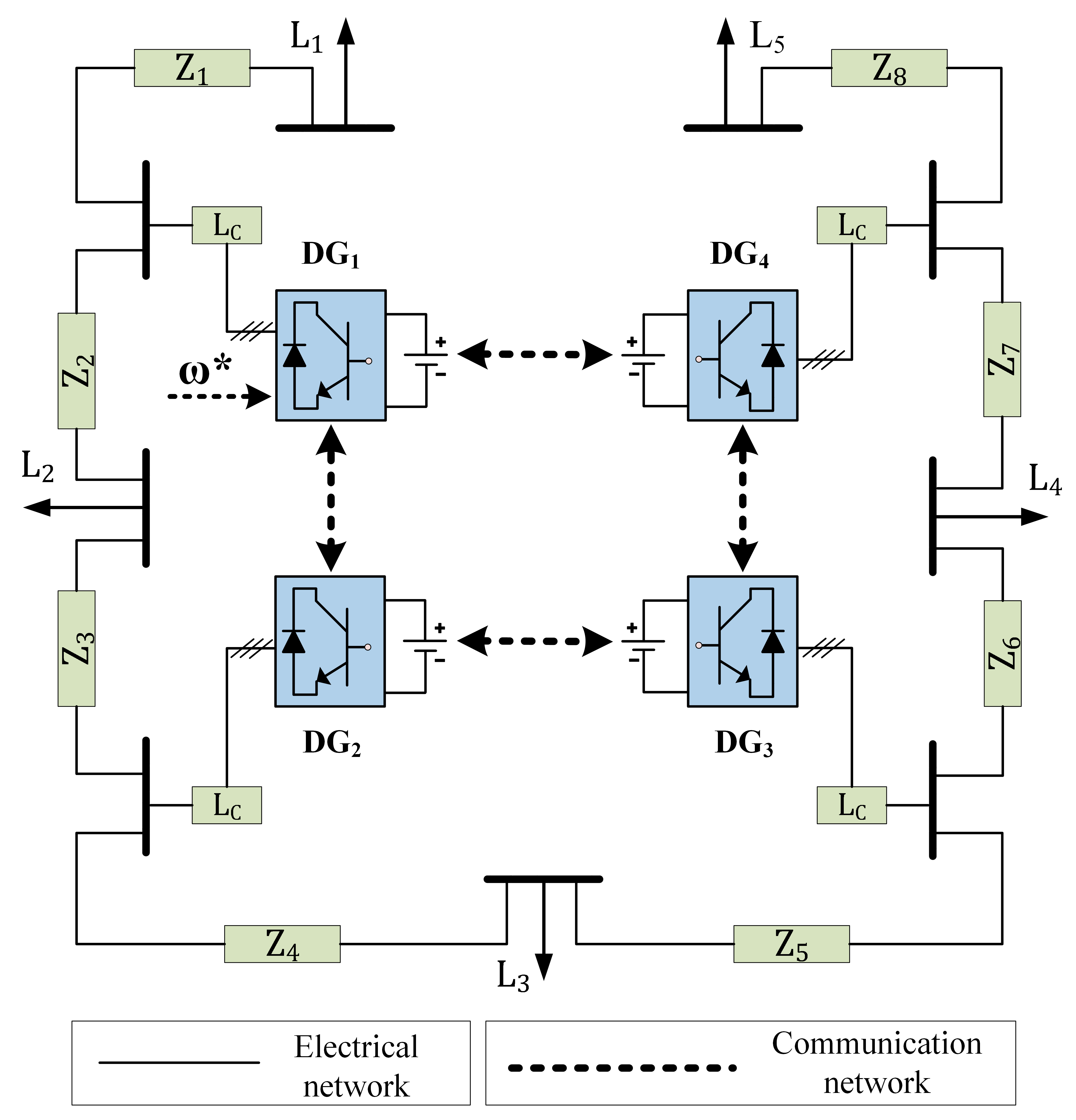}
	\caption{The test islanded AC microgrid under consideration.}
	\label{system}
	\vspace*{-9pt}
\end{figure}

\begin{table}[t]
	\caption{System Parameters}
	\centering
	\begin{tabular}{c c}
		\hline
		Parameter  & Value  \\
		\hline
		Droop gain ($m_{P_1}, m_{P_2}$) & $2 \times 10^{-3}$\\
		
		Droop gain ($m_{P_3}, m_{P_4}$) & $3\times10^{-3}$\\ 
		
		Filter inductance ($L_f$) & $1.35$ mH\\
		
		Filter capacitance ($C_f$) & $50$$\mu$F\\
		
		Coupling inductance ($L_c$) & $0.35$ mH\\
		
		Line impedance ($\rm{Z_1}, \rm{Z_2}$) & $R=1.2\Omega$, $L=$12mH \\
		
		Line impedance ($\rm{Z_3}, \rm{Z_4}$) & $R=1\Omega$, $L=$10mH \\
		
		Line impedance ($\rm{Z_5}, \rm{Z_6}$) & $R=0.8\Omega$, $L=$8mH \\
		
		Line impedance ($\rm{Z_7}, \rm{Z_8}$) & $R=0.5\Omega$, $L=$5mH \\
		
		Loads ($\rm{L_1, L_3, L_5}$) & $P=7$ kW, $Q=4$ kVar\\
		
		Loads ($\rm{L_2, L_4}$) & $P=10$ kW, $Q=5$ kVar\\
		\hline
	\end{tabular}
	\label{table:parameters}
	\vspace*{-10pt}
\end{table}

\section{Simulation Case Studies}\label{Simulation_case_studies}
Consider a 3-phase, 415 V islanded AC microgrid as shown in Fig.~\ref{system}, which comprises $4$ DGs, connected through $8$ transmission lines and $5$ loads, whose values are as listed in Table \ref{table:parameters}. The DGs are connected as per the undirected and connected communication topology, as shown in Fig.~\ref{system}, which has the following Laplacian:
\begin{equation*}
L = \begin{bmatrix}
	2 & -1 & 0 & -1 \\
	-1 & 2 & -1 & 0 \\
	0 & -1 & 2 & -1 \\
	-1 & 0 & -1 & 2
\end{bmatrix}.
\end{equation*}
Assume that DG $1$ acts as the leader and has access to reference frequency $\omega^{*} = 314$~rad/s ($f^{*}=\omega^{*}/2\pi = 50$ Hz), with all the remaining DGs acting as fully connected followers. Therefore, the pinning gain $g_1 = 1$ and $g_i = 0$ for $i = 2, 3, 4$. Associated with each DG, there exist control and auxiliary nodes in the layers $\Sigma$ and $\Pi$, respectively. Consequently, the matrices $A, B, C$ in \eqref{u_w} can be obtained as  
\begin{equation*}
A = \begin{bmatrix}
	-3 & 1 & 0 & 1 \\
	1 & -2 & 1 & 0 \\
	0 & 1 & -2 & 1 \\
	1 & 0 & 1 & -2
\end{bmatrix}, \
B = \begin{bmatrix}
	1\\
	0 \\
	0 \\
	0
\end{bmatrix}, \
C = A\pmb{1}_n = \begin{bmatrix}
	-1\\
	0 \\
	0 \\
	0
\end{bmatrix}.	
\end{equation*}

According to Lemma~\ref{lem_freq_attack_absence} and Lemma~\ref{lem_power_attack_absence}, one can easily verify that, in the absence of any attack, the frequencies of all the DGs converge to $314$ rad/s, while DGs 1, 2 and DGs 3, 4 share equal active power as per their equal droop coefficients. We now analyze protocols \eqref{u_w} and \eqref{u_P} in case of an attack as discussed below.

\begin{figure}[t]
	\centering
	\includegraphics[width=3.3in, height=1.3in]{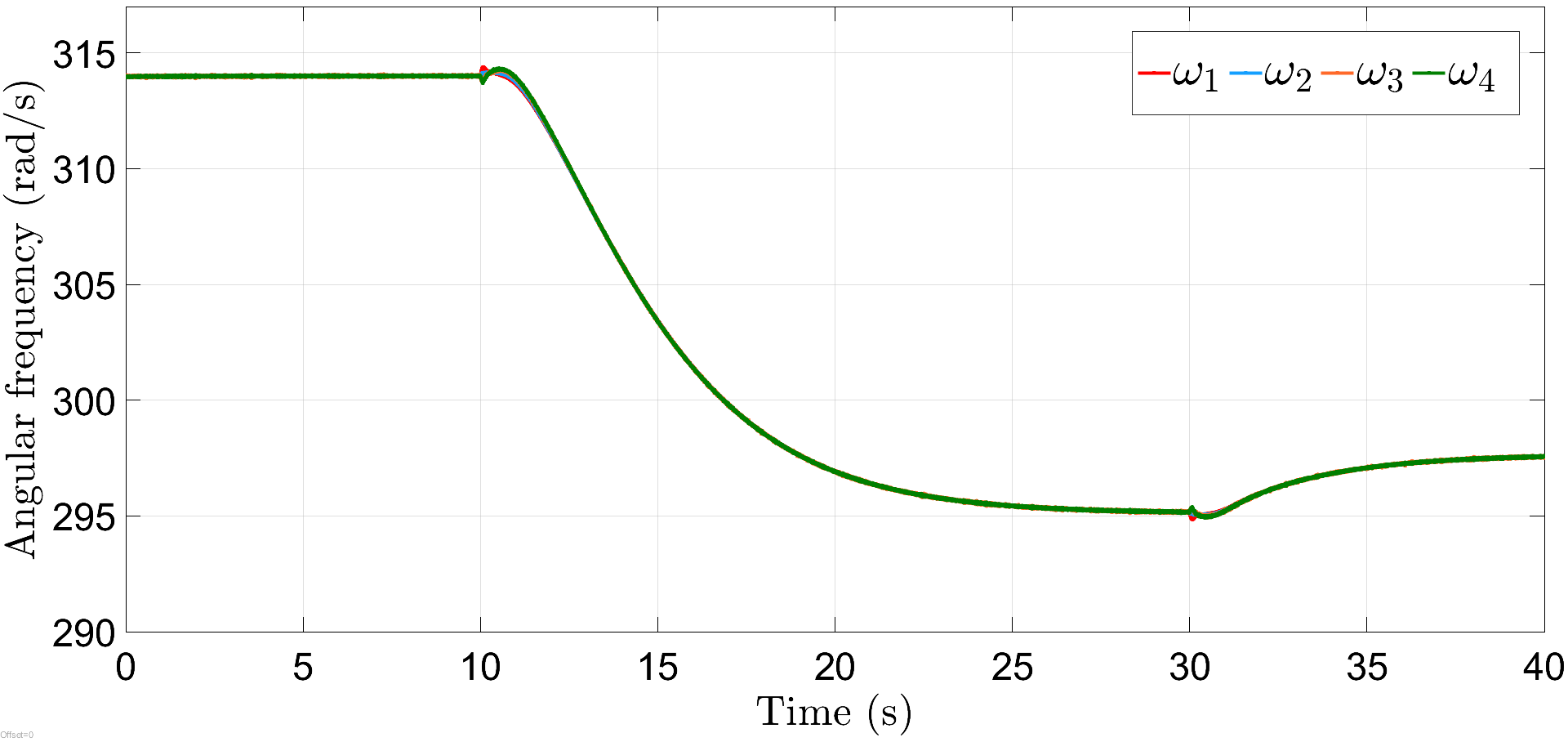}
	\caption{DG frequency under attack and absence of $\Pi$ layer.}
	\label{w0}
	\vspace*{-10pt}
\end{figure}
\begin{figure}[t]
	\centering
	\includegraphics[width=3.3in, height=1.3in]{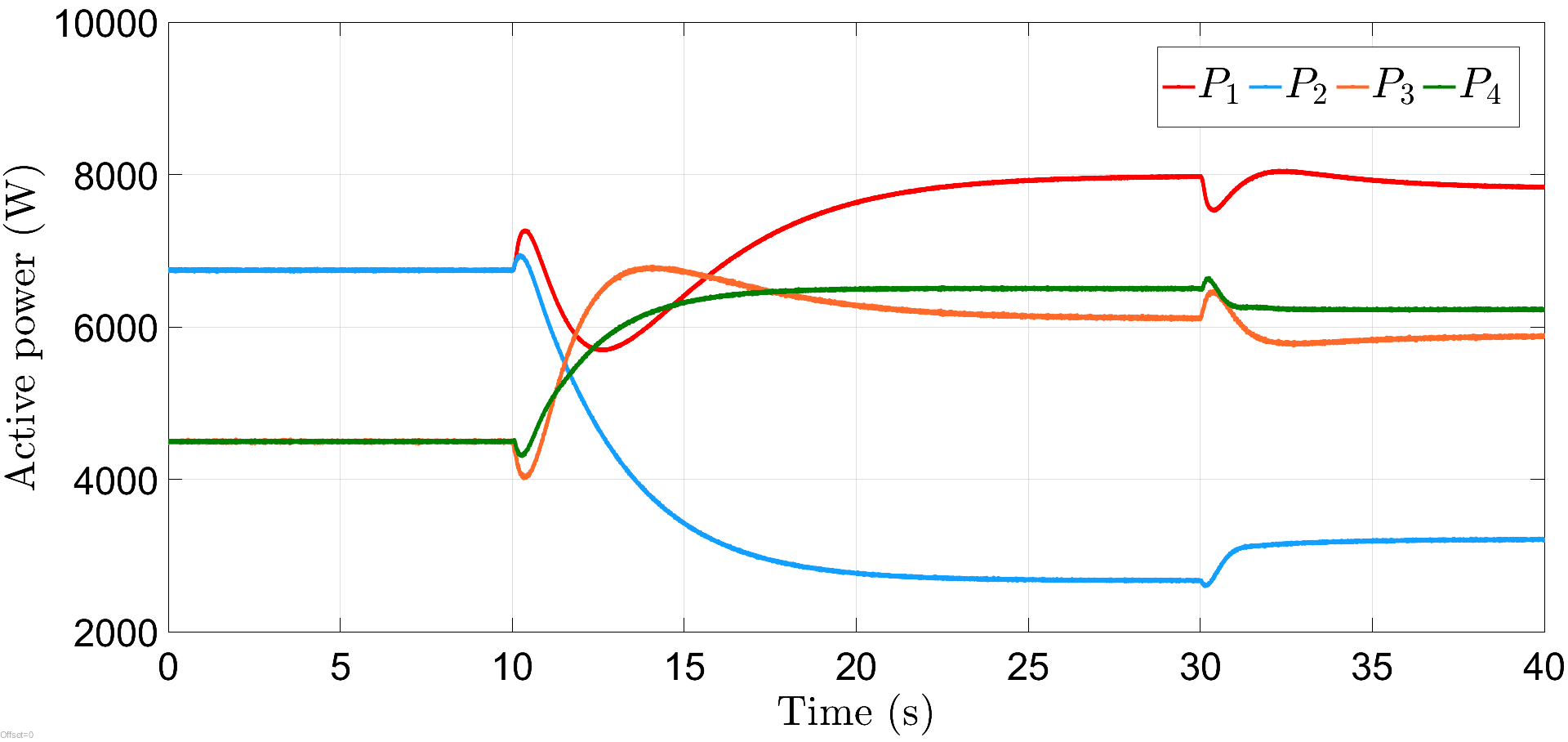}
	\caption{DG active power under attack and absence of $\Pi$ layer.}
	\label{P0}
	\vspace*{-10pt}
\end{figure}

\begin{figure}[t]
	\centering
	\includegraphics[width=3.3in, height=1.3in]{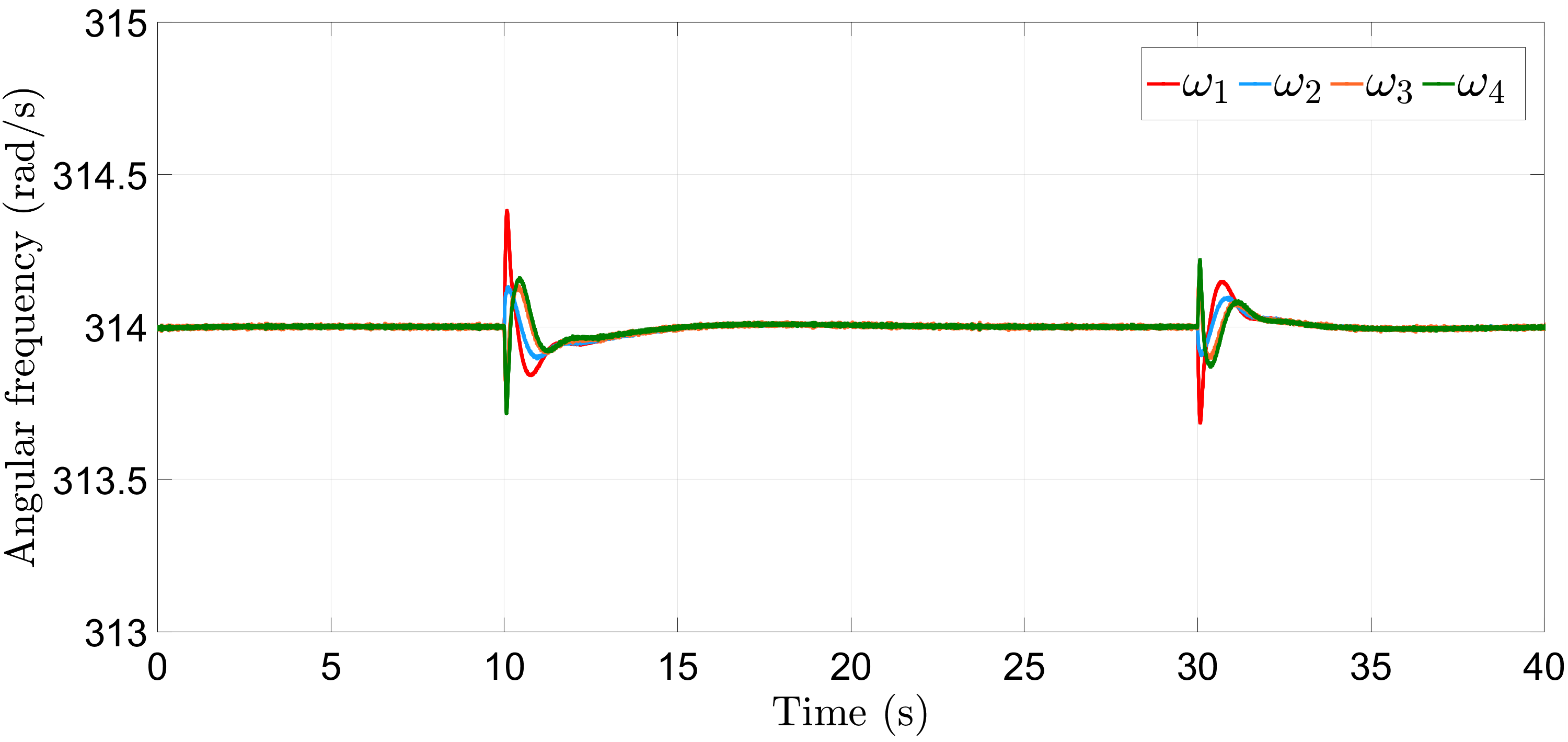}
	\caption{DG frequency under attack and presence of $\Pi$ layer.}
	\label{w1}
	\vspace*{-10pt}
\end{figure}
\begin{figure}[t]
	\centering
	\includegraphics[width=3.3in, height=1.3in]{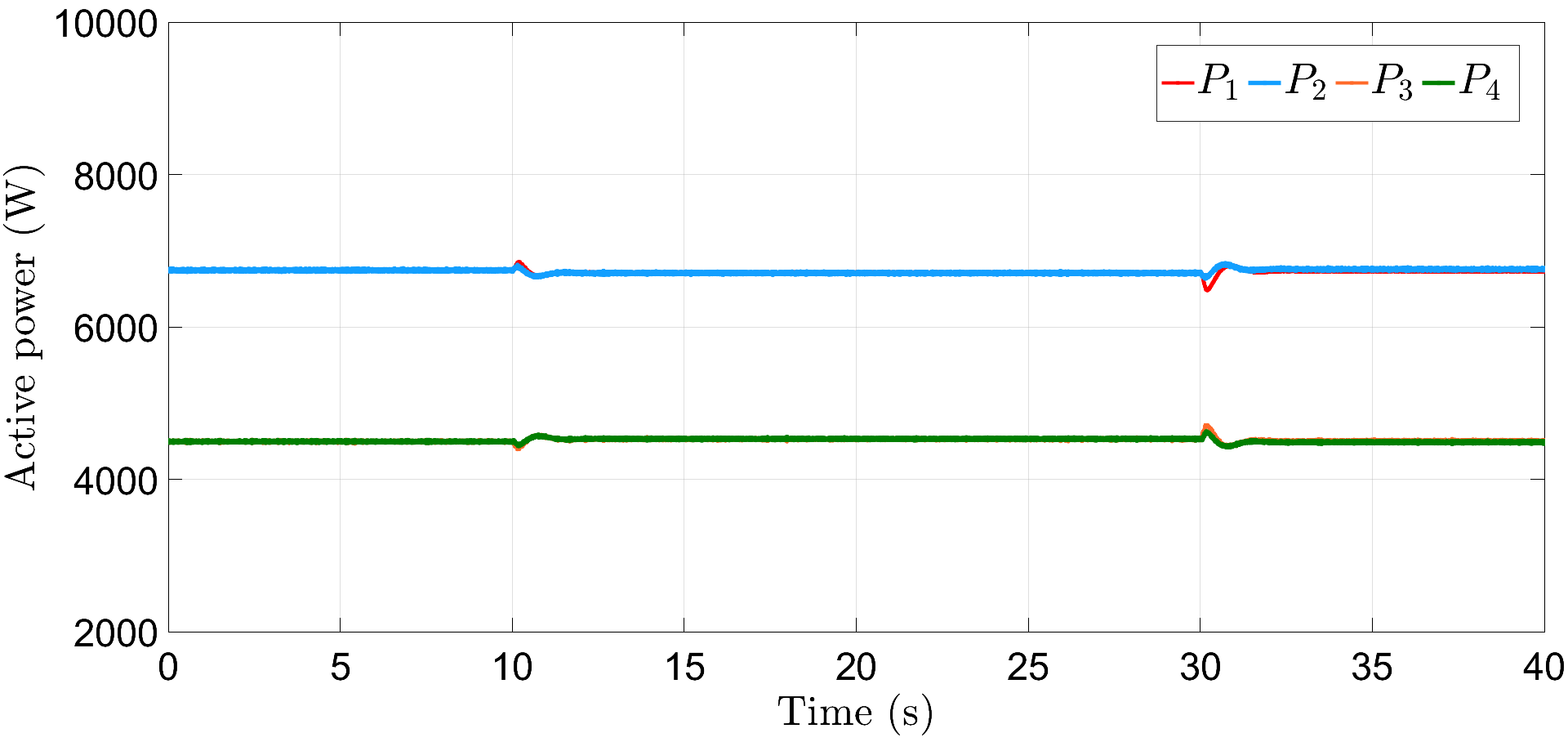}
	\caption{DG active power under attack and presence of $\Pi$ layer.}
	\label{P1}
	\vspace*{-10pt}
\end{figure}

\subsection{Controller performance under attacked condition}\label{controller_performance}
According to Assumption~\ref{assumption}, we consider that the frequency and power attacks are governed by the dynamics: 
\begin{equation} \label{frequency_attack}
\dot{d}_{\omega} = F_{\omega}d_{\omega} + G_{\omega} \omega,
\end{equation}
where 
\begin{equation*}
F_{\omega}= \begin{bmatrix}
	-5 & 0 & 0 & 0\\
	0 & -3 & 0 & 0 \\
	0 & 0 & -5 & 0 \\
	0 & 0 & 0 & -3
\end{bmatrix},
\end{equation*}
\begin{equation*}
G_{\omega} =
 \begin{bmatrix}
	-0.001 & -0.002 & -0.003 & -0.004\\
	-0.003 & -0.001 & -0.004 & -0.002 \\
	0.004 & 0.003 & 0.002 & 0.001\\
	0.002 & 0.004 & 0.001 & 0.003
\end{bmatrix}.
\end{equation*}
and 
\begin{equation}\label{power_attack}
\dot{d}_{P} = F_{p}d_{P} + G_{P} (m_{P}P),
\end{equation}
where
\begin{equation*}
F_{P}= \begin{bmatrix}
	-2.5 & 0 & 0 & 0\\
	0 & -3 & 0 & 0 \\
	0 & 0 & -2.5 & 0 \\
	0 & 0 & 0 & -3
\end{bmatrix},
\end{equation*}
\begin{equation*}
G_{P} = \begin{bmatrix}
	-0.035 & -0.036 & -0.037 & -0.038\\
	-0.088 & -0.085 & -0.086 & -0.087 \\
	0.037 & 0.038 & 0.035 & 0.036 \\
	0.086 & 0.087 & 0.088 & 0.085
\end{bmatrix}.
\end{equation*}
Please note that these frequency and power attacks \eqref{frequency_attack} and \eqref{power_attack} are completely unknown to the controller and are initialized in the system as injections by the attacker at a given time. To emphasize the importance of introducing an auxiliary layer in the distributed control framework, we now discuss the performance of proposed controllers in the absence and presence of $\Pi$, followed by the effect of load perturbations.

\begin{figure}[t]
	\centering
	\includegraphics[width=3.3in, height=1.3in]{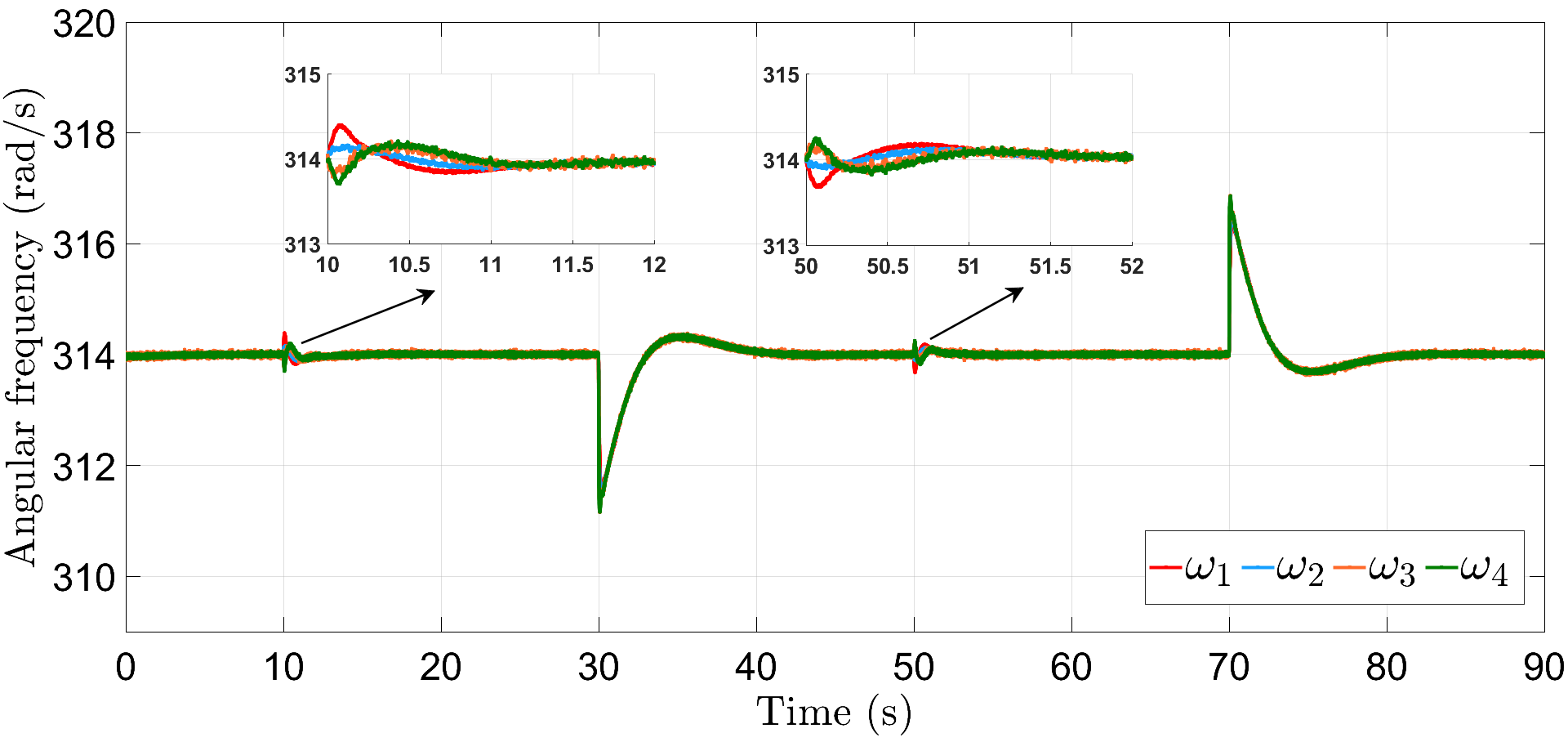}
	\caption{DG frequency under attack and load perturbation.}
	\label{frequency}
	\vspace*{-10pt}
\end{figure}
\begin{figure}[t]
	\centering
	\includegraphics[width=3.3in, height=1.3in]{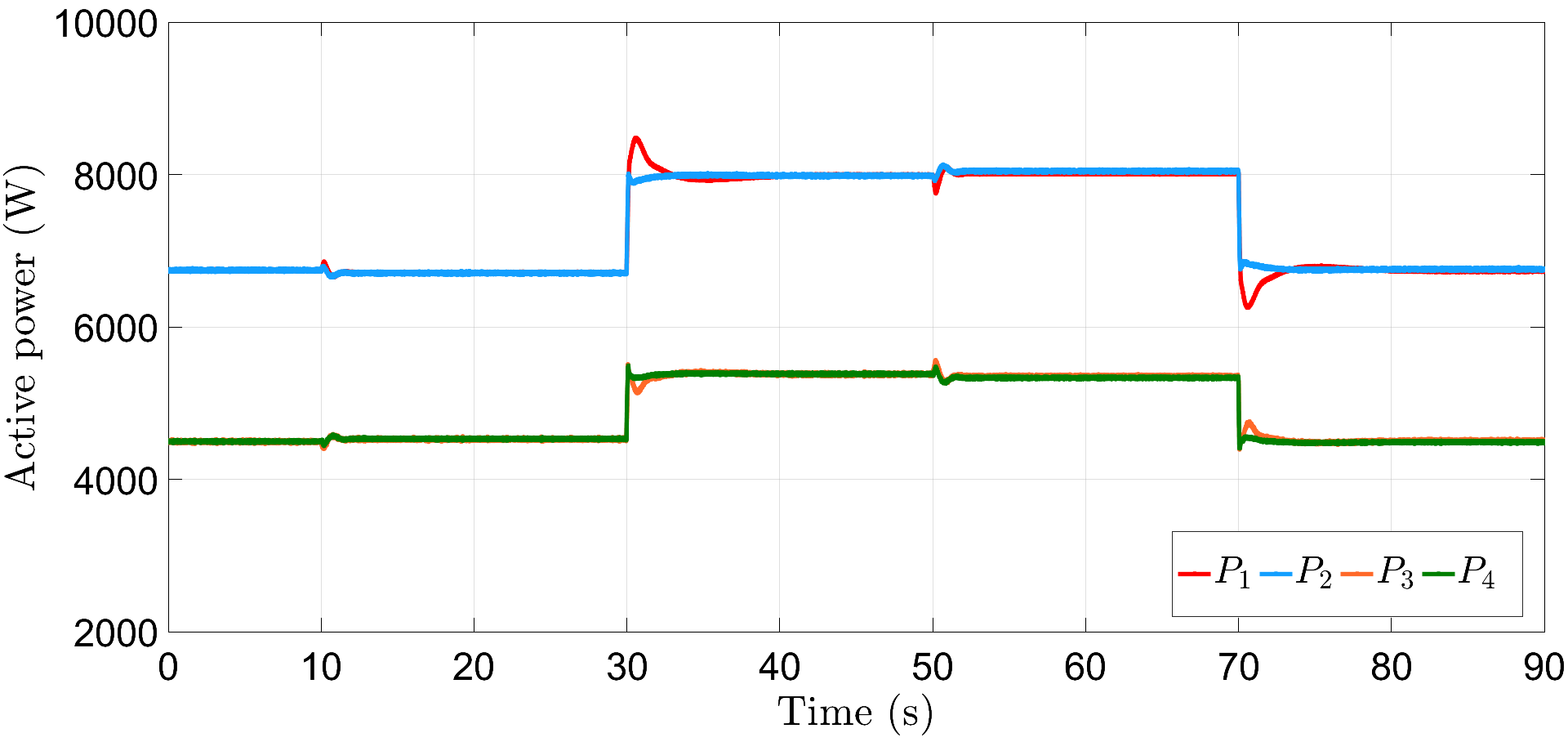}
	\caption{DG active power under attack and load perturbation.}
	\label{power}
	\vspace*{-10pt}
\end{figure}

\begin{figure}[t]
	\centering
	\includegraphics[width=3.3in, height=1.3in]{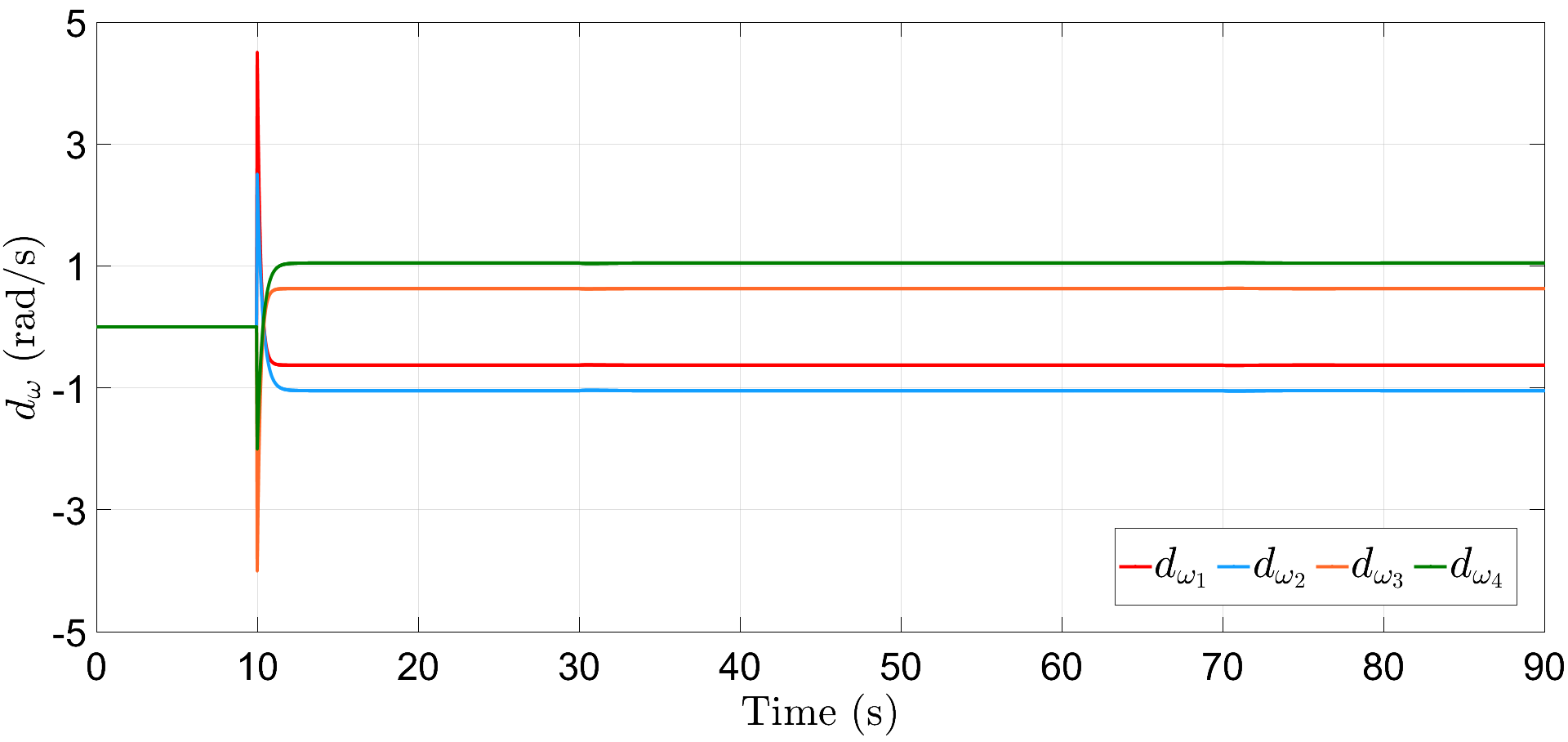}
	\caption{Plot of frequency attack $d_{\omega}$ components with time.}
	\label{frequency_attack_figure}
	\vspace*{-10pt}
\end{figure}

\begin{figure}[t]
	\centering
	\includegraphics[width=3.3in, height=1.3in]{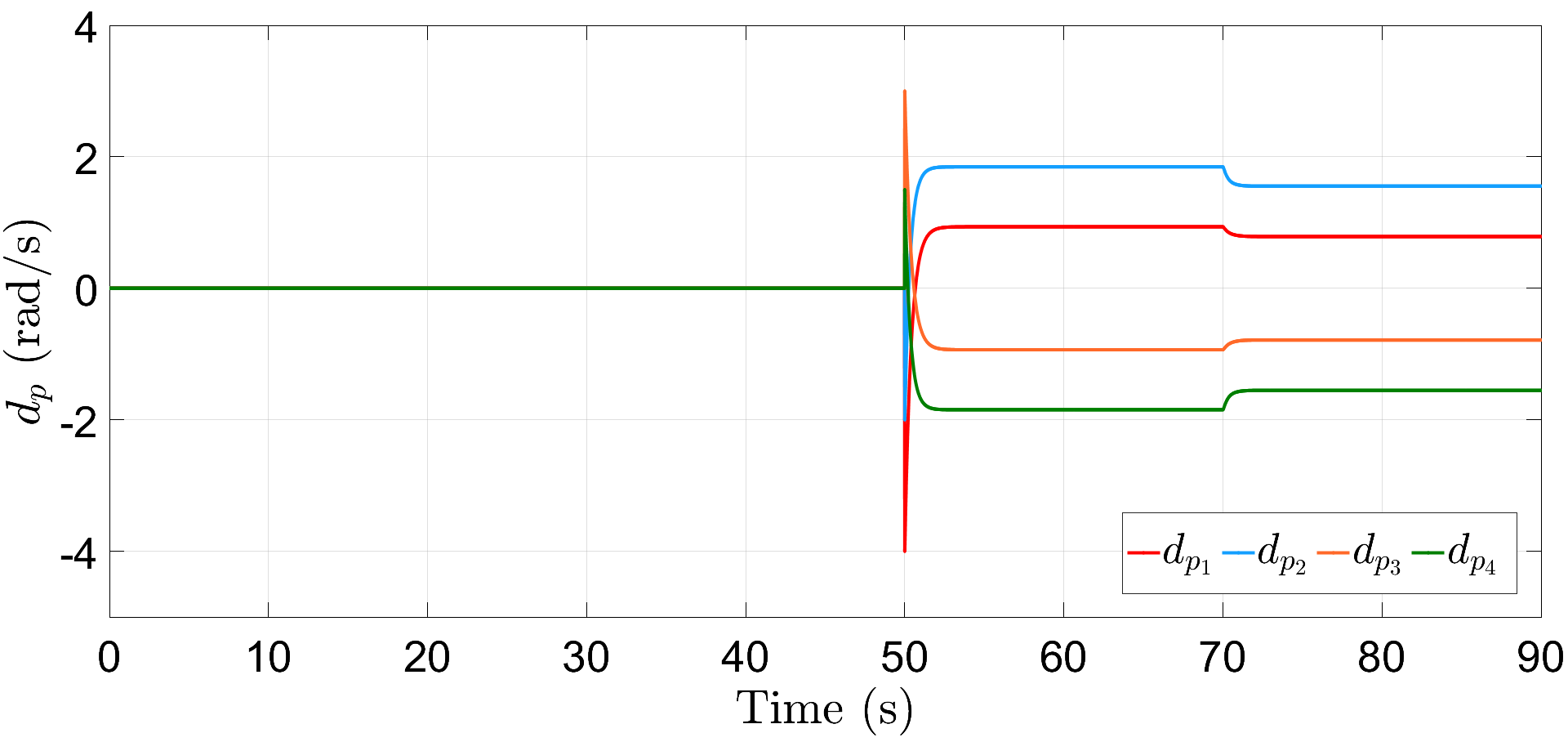}
	\caption{Plot of power attack $d_P$ components with time.}
	\label{power_attack_figure}
	\vspace*{-10pt}
\end{figure}

\subsubsection{Absence of auxiliary layer} We first illustrate the test system in the presence of cyber-attacks and the absence of an auxiliary layer (and hence, the auxiliary nodes) from the frequency and power controllers \eqref{u_w} and \eqref{u_P}, respectively. 

Initially, in the absence of an attack, the microgrid is operating normally under controllers \eqref{u_w_1} and \eqref{u_P_1}, at $314$ rad/s with DGs 1, 2 and 3, 4 delivering equal power at $6.7$ kW and $4.5$ kW, respectively, as shown in Figs.~\ref{w0} and \ref{P0} for the initial 10 s. The frequency and power attacks in \eqref{frequency_attack} and \eqref{power_attack} are initialized at $t=10$ s and $t=30$ s, causing the frequency of all the DGs to deviate by about $19$ rad/s from the nominal value at 10 s and settle down to a common steady-state value of $295$ rad/s due to the distributed action of nodes in the control layer, which is further shifted to $298$ rad/s post introduction of power attack at $t=30$ s. Similarly, the active power-sharing is also adversely affected, as depicted by Fig.~\ref{P0}.

\subsubsection{Presence of auxiliary layer} We now simulate the test microgrid in the presence of auxiliary layer in \eqref{u_w} and \eqref{u_P}, where the initial states $z(0)$ of the auxiliary nodes are randomly chosen and gain $\beta = 2$. As shown in Fig.~\ref{w1}, frequencies are restored in the neighborhood of $314$ rad/s in about $5$ s, post attacks \eqref{frequency_attack} and \eqref{power_attack}, introduced at $t=10$ s and $t=30$ s, respectively. Similarly, Fig.~\ref{P1} shows that power-sharing is also maintained for DGs 1, 2 and 3, 4 in the neighborhood of their nominal values, after a slight deviation at $t=10$ s and $t=30$ s, which lasts approximately about $1.5$ s.


\subsubsection{Effect of load perturbations} 
In addition to the resiliency against frequency and power attacks, we also verify the robustness of proposed controllers \eqref{u_w} and \eqref{u_P} against abrupt load deviations. For illustration, we consider here that the frequency and power attacks are initiated at $t = 10$ s and $t= 50$ s with different (initial) values for each communication link, given by, $d_{\omega}(t = 10~{\rm s}) = [4.5, 2.5, -4, -2]^T$ and $d_P(t = 50~{\rm s})  = [-4, -2.5, 3, 1.5]^T$, respectively. Further, the loads $\rm{L_2, L_4}$ are increased by $0.5$ times their initial magnitude at $t=30$ s and then decreased by the same amount at $t=70$ s, respectively. Fig.~\ref{frequency} depicts that the proposed resilient control scheme restores the DG frequencies to their nominal value within a span of $10$ s. Similarly, Fig.~\ref{power} illustrates accurate power-sharing post load perturbation, with DGs 1, 2 now sharing $8$ kW and DGs 3, 4 sharing $5.4$ kW, respectively. The plots for $d_{\omega}$ and $d_P$ are shown in Figs.~\ref{frequency_attack_figure} and \ref{power_attack_figure}, respectively, where the attacks converge to different steady-state values. Note that the term $G_{\omega}\omega$ in \eqref{frequency_attack} is of the order of $10^{-3}$ and is dominated by the first term $F_{\omega}d_{\omega}$, the plot for $d_{\omega}$ does not show noticeable fluctuations with load perturbations, as compared to $d_P$.    


\subsection{Attack Detection}

\begin{figure}[t!]
	\centering
	\subfigure[During attack]{\includegraphics[scale=0.2]{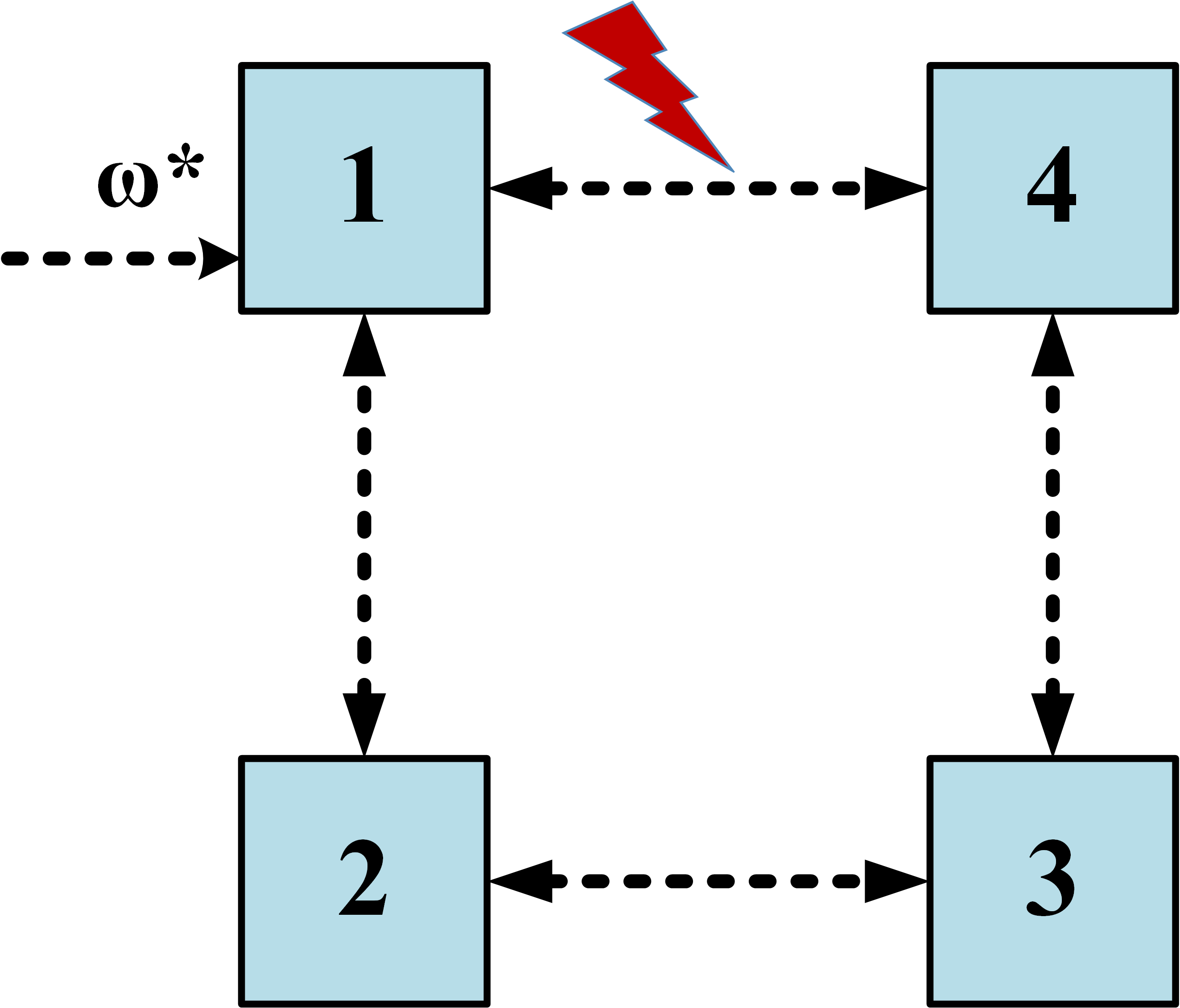}\label{topology_before}} \hspace*{1.0cm}
	\subfigure[After attack]{\includegraphics[scale=0.2]{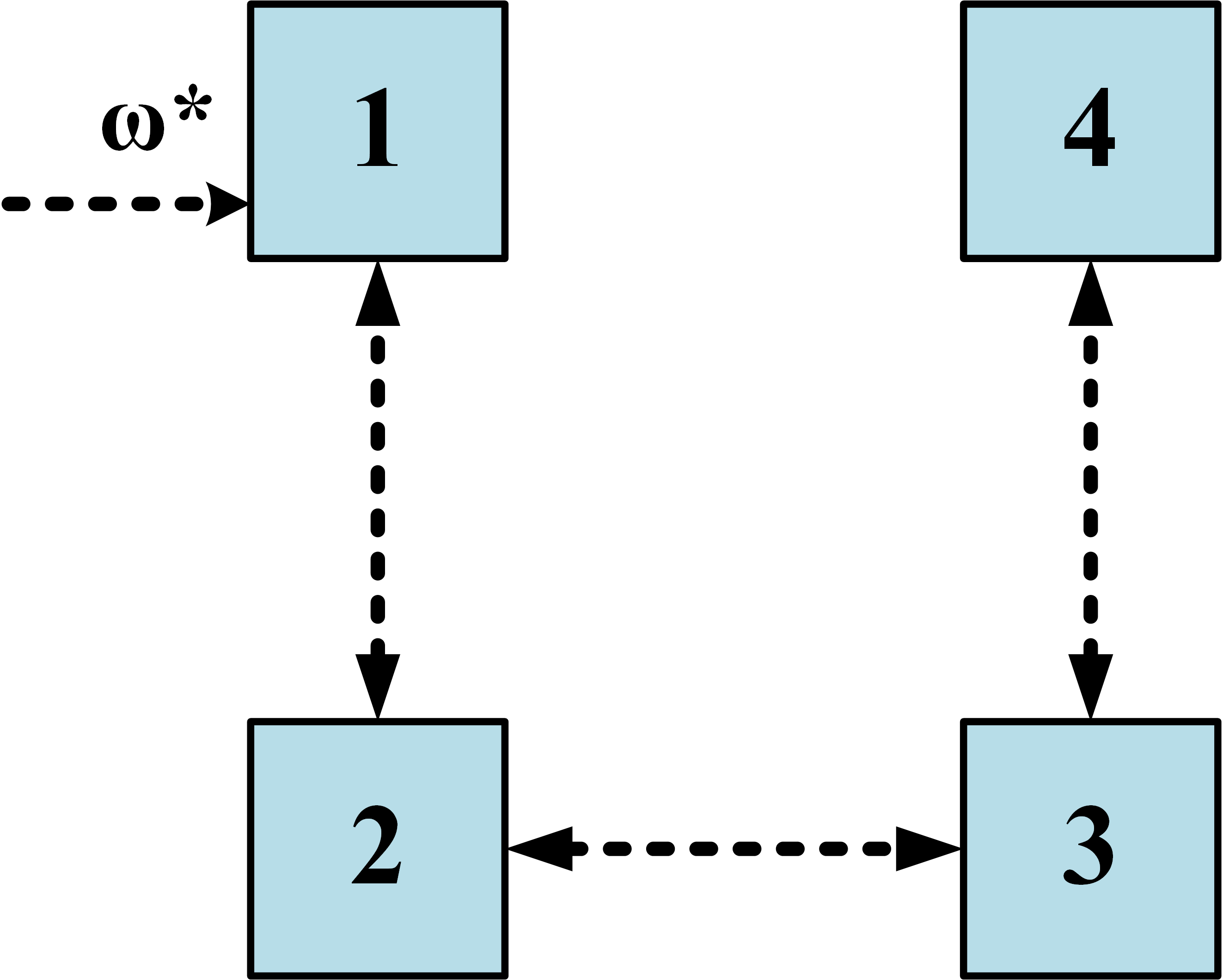}\label{topology_after}}
	\caption{Communication topology before and after attack isolation. Topology (b) also contains a spanning tree.}
	\vspace*{-10pt}
\end{figure}

\begin{figure}[t!]
	\centering
	\includegraphics[width=3.3in, height=1.3in]{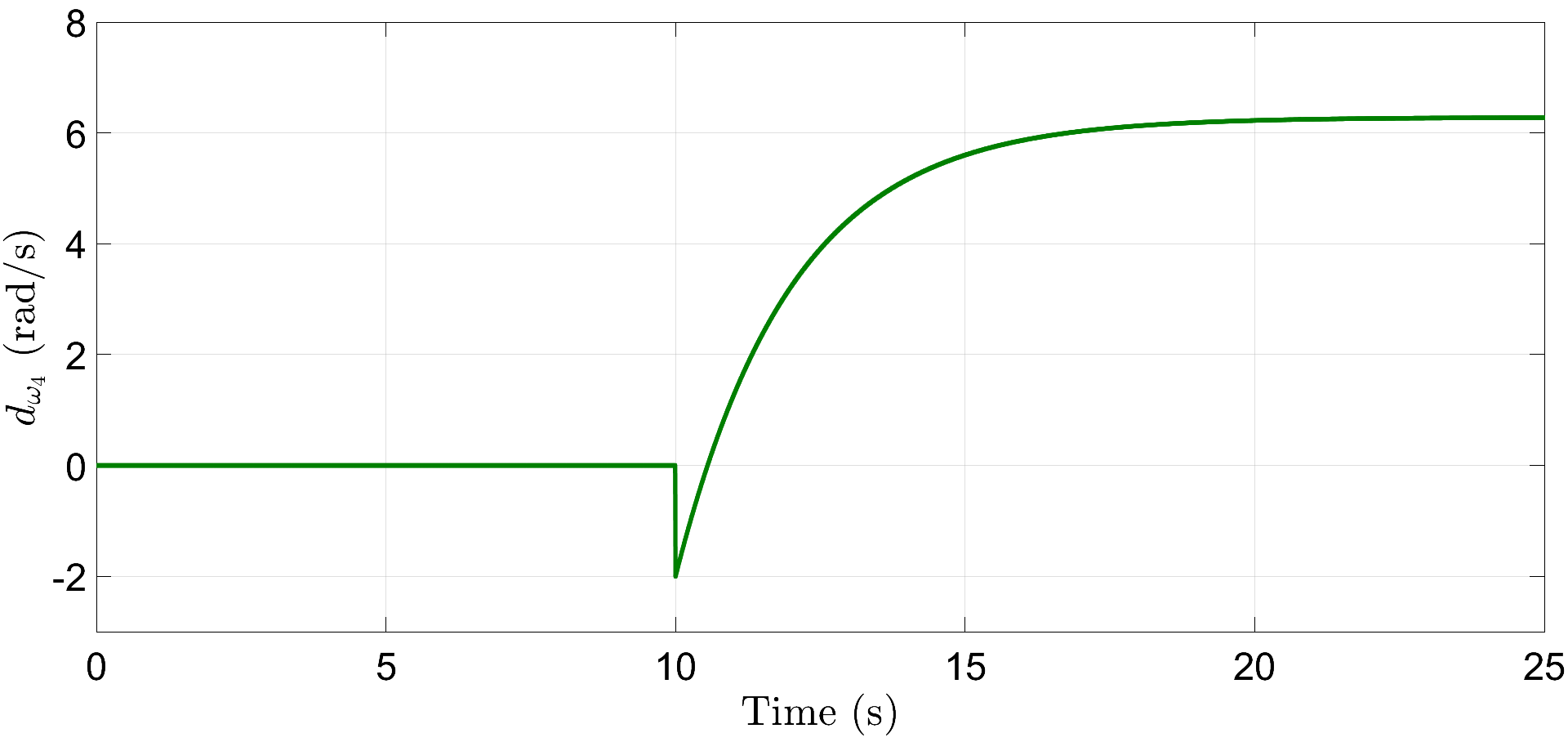}
	\caption{Plot of $d_{\omega_4}$ with time.}
	\label{d_w4}
	\vspace*{-10pt}
\end{figure}
\begin{figure}[t!]
	\centering
	\includegraphics[width=3.3in, height=1.3in]{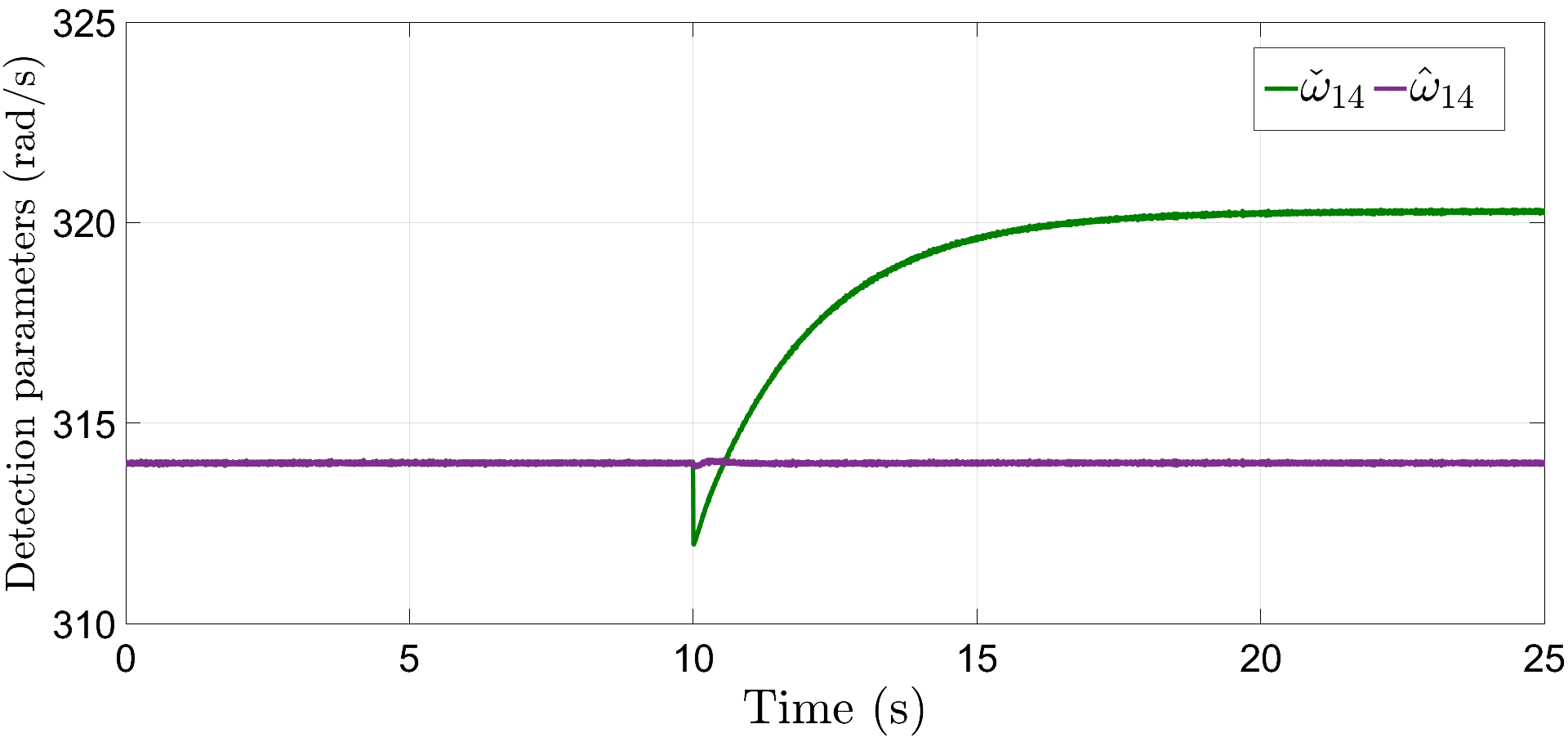}
	\caption{Attack detection at $t = 10$~s as $\check{\omega}_{14} \neq \hat{\omega}_{14}$.}
	\label{detection_figure}
	\vspace*{-10pt}
\end{figure}

\begin{figure}[t!]
	\centering
	\includegraphics[width=3.3in, height=1.3in]{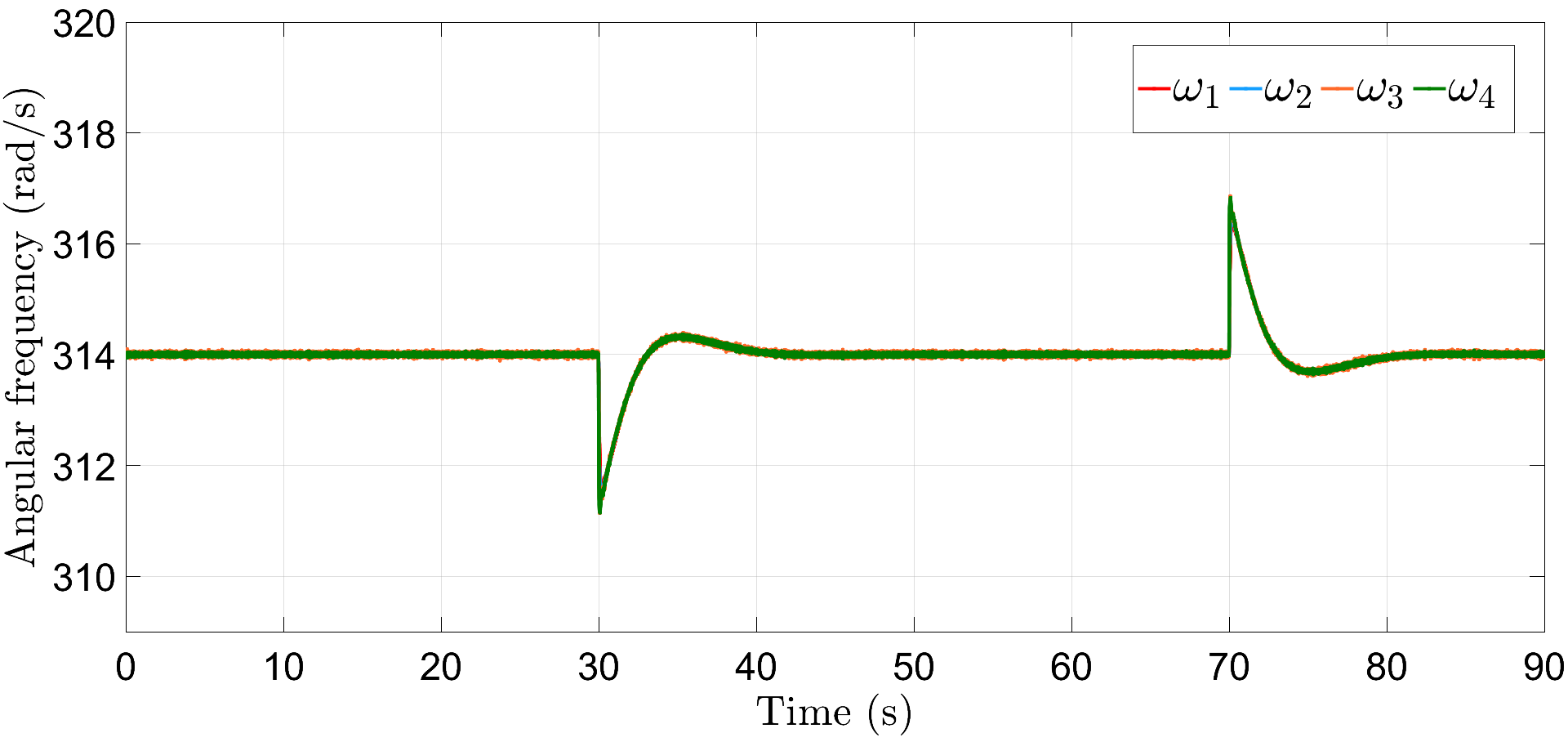}
	\caption{DG frequency with topology in Fig.~\ref{topology_before}.}
	\label{frequency_last}
	\vspace*{-10pt}
\end{figure}
\begin{figure}[t!]
	\centering
	\includegraphics[width=3.3in, height=1.3in]{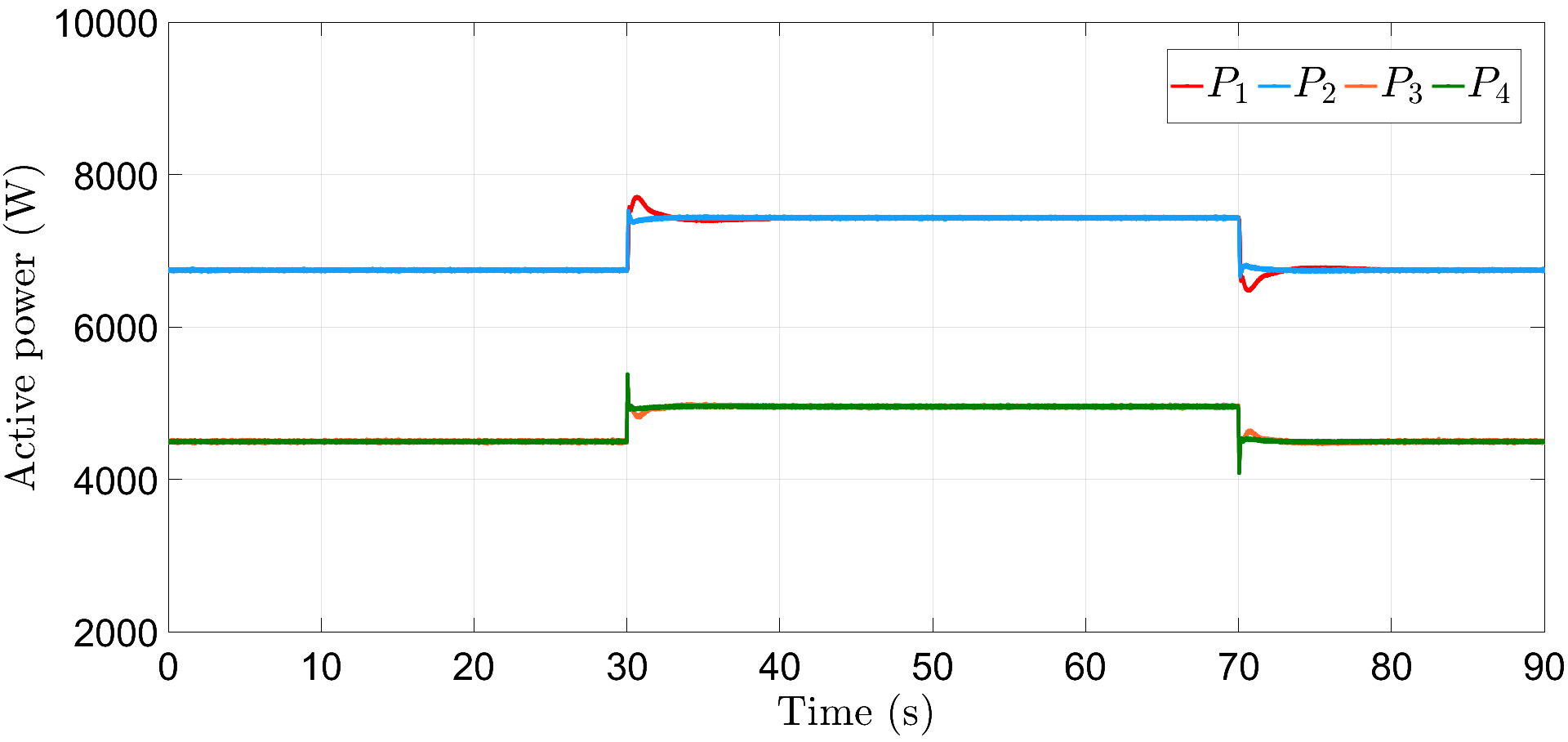}
	\caption{DG active power with topology in Fig.~\ref{topology_after}.}
	\label{power_last}
	\vspace*{-10pt}
\end{figure}

To illustrate this, we consider the following scenario wherein the communication link between DG 1 and DG 4 is under attack $d_{\omega_4} \in \mathbb{R}$ as shown in Fig.~\ref{topology_before}, which is initialized at $t=10$ s with an initial magnitude of $-2$ rad/s as shown in Fig.~\ref{d_w4}, with all the other attack components being 0. Due to this, a corrupted frequency signal $\check{\omega}_{14} = \omega_4 + d_{\omega_4}$ (please refer to \eqref{after_attack}) is sent from node $4$ to node $1$ in the control layer $\Sigma$. However, DGs $1$ and $4$ share an additional information through the auxiliary layer $\Pi$ according to \eqref{detection}, and is given by $\bar{z}_{14} = \beta z_4, \ \bar{\omega}_{14} = z_4 - \beta {\omega}_{4}$, using which, the estimated frequency signal at node 1, due to an attack on node 4, can be calculated using \eqref{test} as $\hat{\omega}_{14} = \frac{1}{\beta} (\frac{\bar{z}_{14}}{\beta} - \bar{\omega}_{14})$.

Fig.~\ref{detection_figure} shows the plot of actual $\check{\omega}_{14}$ and estimated frequency signal $\hat{\omega}_{14}$ with time. Clearly, $\check{\omega}_{14} \neq \hat{\omega}_{14}$ after $t=10$ s, and hence, an attack is detected on the communication link connecting nodes 1 and 4. This attack can be isolated by removing the edge $(1, 4) \in \mathcal{E}$ provided the remaining network is connected, that is, the isolated network contains at least one spanning tree. As shown in Fig.~\ref{topology_after}, since the microgrid system remains connected even after isolation of the corrupted communication link, it operates normally likewise in subsection \ref{controller_performance}; please refer to Figs.~\ref{frequency_last} and \ref{power_last} for frequency restoration and power-sharing in this situation. Here, the load perturbations are introduced with the same magnitude and at similar time instants as in the previous subsection.

\section{Conclusion}\label{Conclusion}
Relying on an auxiliary layer-based design, we investigated an attack-resilient distributed control mechanism for achieving frequency regulation and active power-sharing in an islanded AC microgrid network. One of the main features of our approach is that it not only guarantees the resiliency against simultaneous frequency and power attacks but also, devises an attack detection mechanism as a byproduct, leveraging the interaction between the control and auxiliary layers. The frequency and power controllers were formulated as leader-follower and leaderless multi-agent systems, respectively, accompanied by suitable auxiliary-state dynamics. Under (standard) assumptions on frequency and power attacks, it was shown that the proposed approach assures the frequency restoration and active power-sharing in the small neighborhood of their steady-state values in the presence of any attack. It was proved that these bounds depend on the underlying network topology, the magnitude of the attack signal, and the pinning and control gains decided by the designer. Extensive simulation results were provided to illustrate and verify the theoretical developments in the paper, followed by a simulation demonstrating attack detection and isolation provided the resulting network is connected. 

It would be interesting in future work to extend the analysis to a resilient finite-time distributed framework as the power systems are expected to respond in a certain time to avoid any serious damage.       


\bibliographystyle{IEEEtran}
\bibliography{References}

\end{document}